%% file: main.tex
\documentclass[runningheads,envcountsect,envcountsame]{llncs}
\usepackage[T1]{fontenc}
\usepackage{graphicx}

\usepackage{hyperref}

\usepackage{color}

\usepackage{paralist}
\usepackage[backend=biber,style=lncs]{biblatex}
\bibliography{refs}

\usepackage{amssymb}  
\usepackage{mathtools} 
\usepackage{todonotes} 
\usepackage{comment}
\usepackage{witharrows}
\usepackage{complexity}
\usepackage{accsupp}
\usepackage{subcaption}
\usepackage{mdframed}
\usepackage{tabularx, environ}
\usepackage{enumitem}
\usepackage[capitalize,nameinlink,noabbrev]{cleveref}

\usepackage{thmtools, thm-restate}
\usepackage{xpatch}
\usepackage{apptools}

\makeatletter
\xpatchcmd{\thmt@restatable}%
{\csname #2\@xa\endcsname\ifx\@nx#1\@nx\else[{#1}]\fi}%
{\IfAppendix{\csname #2\@xa\endcsname}{\csname #2\@xa\endcsname\ifx\@nx#1\@nx\else[{#1}]\fi}}%
{}{} %
\makeatother

\usepackage{pgf}
\usetikzlibrary{positioning}    %
\usetikzlibrary{calc}           %
\pgfdeclarelayer{bg}    %
\pgfdeclarelayer{fg}    %
\pgfsetlayers{bg,main,fg}  %

\newcommand{\subparagraphx}[1]{\paragraph{#1}}

\newcommand{\cC}{\ensuremath{\mathcal{C}}}

\newcommand{\pSCClong}{\textsc{Sigma Clique Cover}}
\newcommand{\pSCC}{\textsc{SCC}}
\newcommand{\pCVSlong}{\textsc{Cluster Vertex Splitting}}
\newcommand{\pCVS}{\textsc{CVS}}
\newcommand{\pCEVSlong}{\textsc{Cluster Editing With Vertex Splitting}}
\newcommand{\pCEVS}{\textsc{CEVS}}
\newcommand{\pCElong}{\textsc{Cluster Editing}}
\newcommand{\pCE}{\textsc{CE}}
\newcommand{\pECClong}{\textsc{Edge Clique Cover}}
\newcommand{\pECC}{\textsc{ECC}}

\makeatletter

\newcolumntype{\expand}{}
\long\@namedef{NC@rewrite@\string\expand}{\expandafter\NC@find}

\newcommand{\problempad}{\hspace{0.2em}}

\NewEnviron{problembox}[2][]{%
  \def\problem@arg{#1}%
  \def\problem@framed{framed}%
  \def\problem@lined{lined}%
  \def\problem@doublelined{doublelined}%
  \ifx\problem@arg\@empty%
    \def\problem@hline{}%
  \else%
    \ifx\problem@arg\problem@doublelined%
      \def\problem@hline{\hline\hline}%
    \else%
      \def\problem@hline{\hline}%
    \fi%
  \fi%
  \ifx\problem@arg\problem@framed%
    \def\problem@tablelayout{|>{\bfseries}lX|c}%
    \def\problem@title{\multicolumn{2}{|l|}{%
        \raisebox{-\fboxsep}{\textsc{\large #2}}%
      }}%
  \else
    \def\problem@tablelayout{>{\bfseries}lXc}%
    \def\problem@title{\multicolumn{2}{l}{%
        \raisebox{-\fboxsep}{\textsc{\large #2}}%
      }}%
  \fi%
  \medskip
  \par\noindent%
  \begin{tabularx}{\textwidth}{\expand\problem@tablelayout}%
    \problem@hline%
    \problem@title\\[2\fboxsep]%
    \BODY\\\problem@hline%
  \end{tabularx}%
  \medskip
  \par%
}
\makeatother

\newif\iflong
\newif\ifshort

\longtrue

\iflong
\else
\shorttrue
\fi

\usepackage{etoolbox}
\newcommand{\appsymb}{$\bigstar$}
\newcommand{\appref}[1]{\hyperref[proof:#1]{\appsymb}}

\newcommand{\appendixsection}[1]{%
  \iflong{}\else{}
    \gappto{\appendixText}{\section{Additional Material for \hyperref[#1]{Section~\ref*{#1}} }\label{app:#1}}
  \fi{}
}

\newcommand{\toappendix}[1]{%
  \iflong{}#1\else{}
    \gappto{\appendixText}
    {
        #1
      }
  \fi{}%
}

\newcommand{\appendixproof}[2]{%
  \iflong{}#2\else{}\gappto{\appendixText}
    {
      \subsection{Proof of \cref{#1}}\label{proof:#1}
      #2
    }
  \fi{}
}

\iflong
\makeatletter
\xpatchcmd{\thmt@restatable}%
{\csname #2\@xa\endcsname\ifx\@nx#1\@nx\else[{#1}]\fi}%
{\csname #2\@xa\endcsname}%
{}{} %
\makeatother
\fi

\newcommand{\cupdot}{\mathbin{\mathaccent\cdot\cup}}
\DeclareMathOperator*{\argmin}{\arg\!\min}
\DeclareMathOperator{\wgt}{wgt}
\DeclareMathOperator{\val}{val}
\DeclareMathOperator{\CC}{CC}
\DeclareMathOperator{\cst}{cst}

\newcommand{\N}[1]{\abs{V(#1)}}
\renewcommand{\M}[1]{\abs{E(#1)}}

\newcommand{\set}[1] {
  \mathchoice
  {\left \{ #1 \right \}}
  {\{ #1 \}}
  {\{ #1 \}}
  {\{ #1 \}}
}

\newcommand{\parens}[1] {
  \mathchoice
  {\left ( #1 \right )}
  {( #1 )}
  {( #1 )}
  {( #1 )}
}

\newcommand{\abs}[1] {
  \mathchoice
  {\left | #1 \right |}
  {| #1 |}
  {| #1 |}
  {| #1 |}
}

\definecolor{red}{rgb}{0.984,0.603,0.6}
\definecolor{darkred}{rgb}{0.89,0.102,0.109}
\definecolor{bluex}{rgb}{0.651 0.807 0.89}
\definecolor{darkbluex}{rgb}{0.121,0.47,0.705}
\definecolor{green}{rgb}{0.698,0.874,0.541}
\definecolor{darkgreen}{rgb}{0.2,0.627,0.172}
\definecolor{brown}{rgb}{0.992,0.749,0.435}
\definecolor{darkbrown}{rgb}{0.984375,0.65234375, 0.21875}
\definecolor{gray}{rgb}{0.647,0.647,0.647}
\definecolor{lightgray}{rgb}{0.827,0.827,0.827}
\definecolor{violet}{rgb}{0.792,0.698,0.839}
\definecolor{darkviolet}{rgb}{0.415,0.239,0.603}

\newif\ifshowoutline
\newcommand{\ArrowsC}[3][fleqn,mathindent=0pt,displaystyle,wrap-lines]{%
 \ifshowoutline {\color{bluex}\ \smash{(ArrowsC)} } \fi%
  \begin{equation*}%
    \makebox[\textwidth]{%
            \begin{minipage}[b]{#2}%
                \begin{DispWithArrows*}[#1]%
                  #3%
                \end{DispWithArrows*}%
                \ifshowoutline {\color{bluex}\noindent\rule{\textwidth}{1pt}} \fi%
            \end{minipage}%
    }%
  \end{equation*}%
}

\newcommand{\ArrowsN}[3][displaystyle,wrap-lines]{%
\ifshowoutline {\color{red}\ \smash{(ArrowsN)} } \fi%
\begin{equation*}%
  \makebox[\textwidth]{%
          \begin{minipage}[c]{\textwidth}%
              \begin{DispWithArrows*}[#1]%
                #3%
              \end{DispWithArrows*}%
          \end{minipage}%
  }%
\end{equation*}%
}

\begin{document}

\setlist[description]{font=\normalfont}

\title{The Complexity of\\Cluster Vertex Splitting and Company\thanks{An extended abstract of this work appears in the \emph{Proceeedings of the 49th International Conference on Current Trends in Theory and Practice of Computer Science (SOFSEM~2024)}~\cite{FirbasDHSSVW24-SOFSEM}.}}
\author{
Alexander Firbas
\and
Alexander Dobler\thanks{Supported by the Vienna Science and Technology
Fund (WWTF) under grant 10.47379/ICT19035.}
\and
Fabian Holzer
\and
Jakob Schafellner
\and
Manuel~Sorge\thanks{Partly supported by the Alexander von Humboldt Foundation.}
\and
Anaïs Villedieu
\and
Monika Wißmann
}

\authorrunning{A.\ Firbas et al.}
\institute{TU Wien, Vienna, Austria\\
\email{\{afirbas,adobler,manuel.sorge,avilledieu\}@ac.tuwien.ac.at}}
\maketitle              %
\begin{abstract}
  Clustering a graph when the clusters can overlap can be seen from three different angles: We may look for cliques that cover the edges of the graph with bounded overlap, we may look to add or delete few edges to uncover the cluster structure, or we may split vertices to separate the clusters from each other.
  Splitting a vertex $v$ means to remove it and to add two new copies of $v$ and to make each previous neighbor of $v$ adjacent with at least one of the copies.
  In this work, we study underlying computational problems regarding the three angles to overlapping clusterings, in particular when the overlap is small.
  We show that the above-mentioned covering problem %
  is \NP-complete.
  We then make structural observations that show that the covering viewpoint and the vertex-splitting viewpoint are equivalent, yielding NP-hardness for the vertex-splitting problem.
  On the positive side, we show that splitting at most $k$ vertices to obtain a cluster graph has a problem kernel with $O(k)$ vertices.
  Finally, we observe that combining our hardness results with the so-called critical-clique lemma yields NP-hardness for Cluster Editing with Vertex Splitting, which was previously open (Abu-Khzam et al. [ISCO 2018]) and independently shown to be NP-hard by Arrighi et al. [IPEC 2023].
  We observe that a previous version of the critical-clique lemma was flawed; a corrected version has appeared in the meantime on which our hardness result is based.

  \iflong\keywords{Parameterized algorithms \and Data reduction \and Compact Letter Display \and Computational Complexity} %
  \fi
\end{abstract}

\section{Introduction}\label{sec:intro}
In classical graph-clustering, we want to partition the input graph into clusters that are densely connected, while there are few connections between different clusters.
However, in clusterings of real-world graphs the clusters often overlap~\cite{yang_structure_2014}.
We are interested here in exact algorithms for and complexity of such overlapping clustering problems.
\iflong
  Without overlap, these are well-studied~(e.g.~\cite{DBLP:journals/mst/GrammGHN05,bocker_going_2009,DBLP:journals/mst/ProttiSS09,DBLP:journals/mst/Damaschke10,DBLP:journals/tcs/BodlaenderFHMPR10,DBLP:journals/siamdm/GuoKNU10,DBLP:journals/algorithmica/BockerBK11,DBLP:journals/ipl/BockerD11,DBLP:journals/algorithmica/GuoKKU11,DBLP:journals/jda/Bocker12,komusiewicz_cluster_2012,fomin_tight_2014,marx_fixedparameter_2014,bousquet_multicut_2018,LiPS21}), but less so if we allow overlap~\cite{fellows_graphbased_2011,ArrighiBDSW23,Abu-KhzamBFS21,abu-khzam_cluster_2018}.
\else
  Without overlap, these are well-studied~(see the survey~\cite{CrespelleDFG23}), but less so if we allow overlap~\cite{fellows_graphbased_2011,ArrighiBDSW23,Abu-KhzamBFS21,abu-khzam_cluster_2018}.
\fi
In some applications, clusters may overlap but not very strongly.
We focus mainly on this case.

To understand the complexity, a basic formulation of a clustering with small overlaps can focus on perfect clusterings, i.e., clusters are cliques and all edges of the input graph occur in a cluster.
This leads to the \pSCClong\ (\pSCC) problem, where we seek a covering of the input graph by induced cliques and we want to minimize the total number of times the vertices are covered by the cliques (see \cref{section:scc_np_complete} for a formal definition).\footnote{Note that this is a different optimization goal than the one of the well-studied \pECClong\ problem, where we seek a covering of all edges with a minimum number of induced cliques.}
\pSCC\ was previously studied in the context of displaying information in bioinformatics~\cite{gramm_algorithms_2007} and in combinatorics~\cite{davoodi_edge_2016}.
To our knowledge, its complexity was not known.
We prove that \pSCC\ is \NP-complete (\cref{theorem:scc_np_complete}).%

\looseness=1
An alternative view on overlapping clustering with small overlaps is that of splitting vertices:
A vertex split is a graph operation that takes a vertex $v$ and replaces it by two copies such that the union of the neighborhoods of the copies is equal to the neighborhood of the original vertex~$v$.
Given a graph and an integer~$k$, we may then ask to perform at most $k$ vertex-splitting operations in order to obtain a cluster graph (a disjoint union of cliques).
The cliques in the obtained cluster graph then correspond to the clusters in the original graph.
This yields the \pCVSlong~(\pCVS) problem.
In \cref{section:cvs_np_complete} we show that \pSCC\ and \pCVS\ are indeed equivalent\iflong (see \cref{lemma:cvs_scc_reduction})\fi, and thus both are \NP-complete.
On the positive side, we show that \pCVS\ is fixed-parameter tractable with respect to the number~$k$ of allowed splits, that is, it can be solved in $f(k) \cdot n^{O(1)}$ time where $f$ is a computable function and $n$ the number of vertices.
Indeed, in \cref{chapter:cluster_kernel} we show a stronger result, namely, that \pCVS\ admits an $O(k)$-vertex problem kernel, that is, we may produce with polynomial processing time an equivalent instance that contains $O(k)$ vertices (see \cref{theorem:cluster_kernel}).
This result relies on an analysis of the structure of the so-called critical cliques of the input graph.
Informally, a critical clique is an induced clique in the input graph with vertex set $C$ such that all vertices in $C$ have pairwise the same neighbors outside of $C$ and such that there is no critical clique that strictly contains~$C$.\footnote{Alternatively, a critical clique is a maximal set of pairwise true twins.}
The \pCEVSlong~(\pCEVS) problem~\cite{abu-khzam_cluster_2018} is closely related to the above two problems.
The difference is that the underlying clustering model allows the clusters to be imperfect, that is, the clusters may miss a small number of edges and there may be a small number of edges that are not contained in any cluster.
More precisely, in \pCEVS\ we are given a graph~$G$ and an integer $k$ and we want to obtain a cluster graph from $G$ by at most $k$ modifications.
As modifications we are allowed to split vertices and to add or delete edges.
It was previously open whether \pCEVS\ is NP-hard~\cite{abu-khzam_cluster_2018} which has been independently and in parallel to our work been shown to be true~\cite{ArrighiBDSW23}.
Our impetus was to show NP-hardness of \pCEVS, too, and, indeed, combining our NP-hardness result for \pSCC\ with a so-called critical-clique lemma~\cite{abu-khzam_cluster_2018,abu-khzam_cluster_2019v1} yields NP-hardness of \pCEVS\ (see \cref{sec:cevs}).
We refrained from publishing this result at first, because the critical-clique lemma as stated by Abu-Khzam et al.~\cite{abu-khzam_cluster_2018,abu-khzam_cluster_2019v1} and used in references~\cite{ArrighiBDSW23,askeland_overlapping_2022} is incorrect, see the counterexample in \cref{sec:ccl}.
Fortunately, after the appearance of our counterexample, a corrected variant of the critical-clique lemma appeared~\cite{abu-khzam_cluster_2023v2}, completing our alternative NP-hardness proof of \pCEVS.

\toappendix{
\iflong  
  \subparagraphx{Related work}
\else
  \section{Further Related Work}
\fi

The problems we study are related to two problems with similar context but that correspond to clusterings without overlap.
First, there is the well-researched \pCElong~(\pCE) problem, in which we want to add or delete a minimum number of edges in a given graph to obtain a cluster graph~\cite{DBLP:journals/mst/GrammGHN05,bocker_going_2009,DBLP:journals/mst/ProttiSS09,DBLP:journals/mst/Damaschke10,DBLP:journals/tcs/BodlaenderFHMPR10,DBLP:journals/siamdm/GuoKNU10,DBLP:journals/algorithmica/BockerBK11,DBLP:journals/ipl/BockerD11,fellows_graphbased_2011,DBLP:journals/algorithmica/GuoKKU11,DBLP:journals/jda/Bocker12,komusiewicz_cluster_2012,fomin_tight_2014,marx_fixedparameter_2014,bousquet_multicut_2018,LiPS21}.
For instance, it is known that \pCE\ is \NP-hard, fixed-parameter tractable, and admits a $2k$-vertex problem kernel.
\pCE\ is one of a broad range of so-called edge-modification problems, see Crespelle et al.~\cite{CrespelleDFG23} for a recent survey.

Second, we have \pECClong~(\pECC), wherein we look for covering all edges of a graph with at most some given number $s$ of induced cliques.
Here, it is known that covering all edges of a given graph with at most $s$ induced cliques can be done in $2^{O(4^s)} + n^{O(1)}$ time~\cite{gramm_data_2009}, but not substantially faster than that~\cite{cygan_known_2016}.

\pCE\ has been extended to a variant modeling overlapping clustering~\cite{fellows_graphbased_2011}, where, instead of trying to get a cluster graph, we modify the edges to obtain a graph in which at most a bounded number of maximal cliques overlap in each vertex.
If we can split a bounded number of vertices to obtain a cluster graph, then the input graph indeed has such a bounded-overlap property, but not vice versa.

Vertex splitting as a graph operation has appeared also in other contexts~\cite{tension_free_layouts,eppstein2018planar,planar_splitting,planarizing_vs_fpt}, such as splitting vertices towards obtaining a planar graph.
Systematic investigation into the complexity of vertex-splitting towards obtaining a fixed graph property began only recently~\cite{firbas_establishing_2023,baumann2023parameterized}.
}                               %

\iflong
\subparagraphx{Organization}
We will establish the following chain of polynomial-time reductions, based on the classical \NP-hard \textsc{Node Clique Cover (NCC)} problem~\cite{karp}:
\begin{align*}
    \textsc{Node Clique Cover} &\leq_\P \textsc{Sigma Clique Cover}\\
    &\leq_\P \textsc{Cluster Vertex Splitting}\\
    &\leq_\P \textsc{Cluster Editing With Vertex Splitting}.
\end{align*}
We give the first reduction in \cref{section:scc_np_complete}, the second in \cref{section:cvs_np_complete}, and the last in \cref{sec:cevs}.
The informal definitions of these problems have been given above, the formal definitions will be given in the corresponding sections.
The problem kernel is shown in \cref{chapter:cluster_kernel} and the critical-clique lemma is treated in \cref{sec:ccl}.
\else

\fi

\ifshort
\smallskip
\noindent \textit{Due to space constraints, statements marked with}~\appsymb~\textit{are proved in the arXiv version of this paper \cite{DBLP:journals/corr/abs-2309-00504}}.%
\fi%
\section{Preliminaries}
For a positive integer~$n \in \mathbb{N}$ we use $[n]$ to denote $\{1, 2, \ldots, n\}$. %
For a set $X$, we denote by $\mathcal{P}(X)$ its power set.
Moreover, for a family of sets $\mathcal{X}$, we write $\bigcup \mathcal{X}$ for the union of all sets member of $\mathcal{X}$, that is, $\bigcup_{X \in \mathcal{X}} X$.
We denote disjoint unions by $\cupdot$.
Unless explicitly mentioned otherwise, all graphs are undirected and without parallel edges or self-loops.
Given a graph~$G$ with vertex set $V(G)$ and edge set $E(G)$, we denote the neighborhood of a vertex $v\in V(G)$ by $N_G(v)$.
If the graph~$G$ is clear from the context, we omit the subscript~$G$.
For $V'\subset V(G)$, we write $G[V']$ for the graph induced by the vertices $V'$.
For $u,v\in V(G)$ we write $uv$ as a shorthand for $\{u,v\}$, $G-v$ for $G[V\setminus \{v\}]$, and $d_G(v)$ for $|N_G(v)|$.
The graph $K_n$ is the complete graph on $n$ vertices. We write $G\simeq H$ if a graph $G$ is isomorphic to~$H$.
A \emph{cluster graph} is a graph in which every connected component is a clique.
Equivalently, a cluster graph does not contain a path $P_3$ with three vertices as an induced subgraph.
A \emph{vertex split} operation applied to a graph $G=(V,E)$ and $u\in V$ results in a graph $G'=(V',E')$ such that $V'=V\setminus \{u\} \cup \{v,w\}$ with $v,w\not\in V$, and $E'$ is obtained from $E$ by making each vertex adjacent to $u$ adjacent to at least one of $v$ and $w$; that is, $N_{G'}(v)\cup N_{G'}(w)=N_{G}(u)$.

\toappendix{
  \ifshort
    \section{Further Preliminaries}
  \fi
    
Some of our results are in terms of parameterized complexity~\cite{fpt_downey_fellows,fptbook,FlumG06,Niedermeier06}.
Briefly and informally, in a \emph{parameterized problem}, each instance $x \in \Sigma^{*}$ is equipped with a parameter~$k$.
Such a problem is \emph{fixed-parameter tractable} if it can be solved in $f(k) \cdot n^{O(1)}$ time, where $f$ is a computable function and $n$ the input size.
A parameterized problem has a \emph{problem kernel} if there is a polynomial-time self-reduction such that in the resulting instances the size is bounded by $g(k)$, where $g$ is a computable function and $k$ is the parameter.
The function $g$ is also called the \emph{size} of the problem kernel.
}

\section{NP-Completeness of \textsc{Sigma Clique Cover}}
\label{section:scc_np_complete}
\appendixsection{section:scc_np_complete}
To start, we will fix some notation.
Leading up to the formulation of the sigma clique cover problem, we first define the notion of a sigma clique cover:
\begin{definition}\label{definition:scc}
  Let $G$ be a graph. Then, $\mathcal{C} \subseteq \mathcal{P}(V)$ is called a \emph{sigma clique cover} of 
  $G$ if
  \iflong    
    \begin{enumerate}
    \item $G[C]$ is a clique for all $C \in \mathcal{C}$ and
    \item for each $e \in E(G)$, there is $C \in \mathcal{C}$ such that $e \in E(G[C])$, that is, all edges of $G$ are ``covered'' by some clique of $\mathcal{C}$.
    \end{enumerate}
    The \emph{weight} of a sigma clique cover $\mathcal{C}$ is denoted by $\wgt(\mathcal{C})$, where
    \[
      \wgt(\mathcal{C}) \coloneqq \sum_{C \in \mathcal{C}} |C|.    
    \]  

  \else
    \begin{inparaenum}
    \item $G[C]$ is a clique for all $C \in \mathcal{C}$ and
    \item for each $e \in E(G)$, there is $C \in \mathcal{C}$ such that $e \in E(G[C])$, that is, all edges of $G$ are ``covered'' by some clique of $\mathcal{C}$.
    \end{inparaenum}
    The \emph{weight} of a sigma clique cover $\mathcal{C}$ is $\wgt(\mathcal{C}) \coloneqq \sum_{C \in \mathcal{C}} |C|$.
  \fi      
\end{definition}

Now, we can formulate the associated decision problem:
\begin{problembox}[framed]{\textsc{\problempad Sigma Clique Cover (SCC)}}
    \problempad Input: & A tuple $(G, s)$, where $G$ is a graph and $s \in \mathbb{N}$. \\
    \problempad Question: & Is there a sigma clique cover $\mathcal{C}$ of $G$ with $\wgt(\mathcal{C}) \leq s$?
\end{problembox}

Note that \pSCC{} is not equivalent to the well-studied \pECClong{} problem, whose optimization goal is to minimize $|\mathcal{C}|$ rather than $\wgt(\mathcal{C})$.
To show that \textsc{SCC} is \NP-hard, we reduce from the \textsc{Node Clique Cover} problem.
Analogous to the case of \textsc{SCC}, to define said problem formally, we first need introduce the notion of a node clique cover:
\begin{definition}\label{definition:ncc}
  Let $G$ be a graph. Then, $\mathcal{C} \subseteq \mathcal{P}(V)$ is called a \emph{node clique cover} of 
  $G$ if
  \iflong
    \begin{enumerate}
    \item $G[C]$ is a clique for all $C \in \mathcal{C}$ and
    \item for each $v \in V(G)$, there is $C \in \mathcal{C}$ such that $v \in V(G[C])$, that is, all vertices of $G$ are ``covered'' by some clique  $C \in \mathcal{C}$.
    \end{enumerate}
    The \emph{size} of a node clique cover $\mathcal{C}$ is denoted by $|\mathcal{C}|$.
  \else
    \begin{inparaenum}
    \item $G[C]$ is a clique for all $C \in \mathcal{C}$ and
    \item for each $v \in V(G)$, there is $C \in \mathcal{C}$ such that $v \in V(G[C])$, that is, all vertices of $G$ are ``covered'' by some clique  $C \in \mathcal{C}$.
    \end{inparaenum}
    The \emph{size} of a node clique cover $\mathcal{C}$ is $|\mathcal{C}|$.
  \fi
\end{definition}

With this, we can formulate the \NP-hard~\cite{karp} \textsc{Node Clique Cover} problem:
\begin{problembox}[framed]{\textsc{\problempad Node Clique Cover (NCC)}}
    \problempad Input: & A tuple $(G, k)$, where $G$ is a graph and $k \in \mathbb{N}$. \\
    \problempad Question: & Is there a node clique cover $\mathcal{C}$ of $G$ with $|\mathcal{C}| \leq k$?
\end{problembox}

\toappendix{
\begin{figure}[t]
    \begin{center}
        \includegraphics{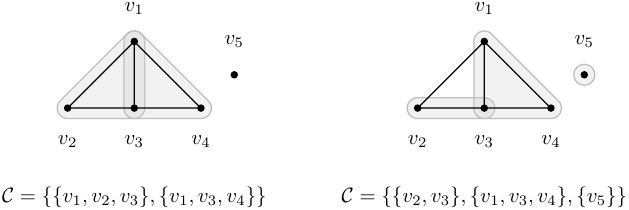}
    \end{center}
    \caption[Example illustrating the differences between \textsc{NCC} and \textsc{SCC}.]{A graph with its unique minimum-weight sigma clique cover (left) and one of its multiple minimum-cardinality node clique covers (right).}
    \label{figure:scc_and_ncc}
\end{figure}

Note that the \textsc{SCC} and \textsc{NCC} problem are similar on a superficial level, but differ in two core aspects:
Firstly, the notion of a sigma clique cover mandates that all edges be covered, in comparison to node clique covers, where all vertices need to be covered,
and secondly, the ``difficulty'' of the \textsc{SCC} problem lies in minimizing a cumulative weight, in comparison to the \textsc{NCC} problem, where it is the number of cliques to be minimized.
See \cref{figure:scc_and_ncc} for a contrasting example.
}

To formulate our reduction from \textsc{NCC} to \textsc{SCC}, we introduce notation to extend a graph with independent universal vertices.
\iflong
  See \cref{figure:k3_universal} for an example of \cref{definition:universal_node}.
\else
\fi%
\begin{definition}\label{definition:universal_node}
    Let $G=(V,E)$ be a graph and $\ell \in \mathbb{N}$. Using a set $\{u_1, \dots, u_{\ell}\}$ of $\ell$ new vertices called \emph{universal vertices}, we construct a new graph $G^\ell$ with
    \[
    G^{\ell} \coloneqq  \left ( V \cup \{u_1, \dots, u_{\ell}\}, E \cup \{u_iv \mid 1 \leq i \leq \ell, v \in V\} \right ).
    \]
\end{definition}

\toappendix{
\begin{figure}[t]
    \begin{center}
        \includegraphics{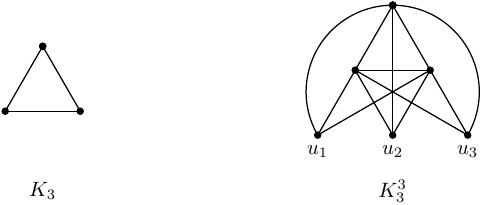}
    \end{center}
    \caption[A graph with universal nodes, illustrating \cref{definition:universal_node}.]{$K_3$ and $K_3^3$, illustrating \cref{definition:universal_node}.}
    \label{figure:k3_universal}
\end{figure}
}

Note that universal vertices themselves are not adjacent to each other.
Informally, the main intuition behind our reduction from \textsc{NCC} is to add a sufficient number of universal vertices to the instances of \text{NCC},
such that concerning the derived instances of \textsc{SCC}, it will be ``combinatorially favorable'' to select cliques that contain a universal vertex.
\begin{lemma}\label{lemma:scc_ncc_reduction}
    Let $G = (V,E)$ be a graph and $\ell \coloneqq 2 |E| + 1$. Then,
    $(G,s)$ is a positive instance of \textsc{NCC} if and only if
    $\left (G^{\ell},\ell \left ( |V| + s + 1 \right ) - 1 \right )$
    is a positive instance of \textsc{SCC}.
\end{lemma}%
\begin{figure}[t]%
    \begin{center}%
        \includegraphics{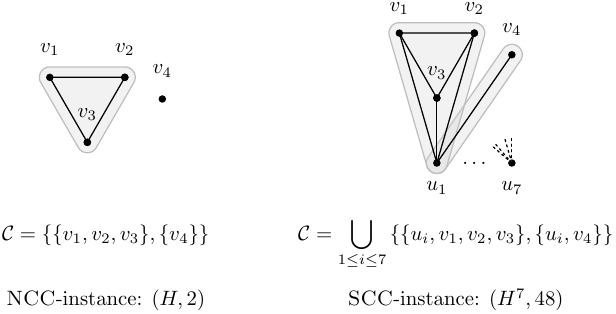}%
    \end{center}%
    \caption[Example for the reduction from \textsc{NCC} to \textsc{SCC}.]{Example for our reduction from \textsc{NCC} to \textsc{SCC}.
      \iflong
        On the left, we see a \textsc{NCC}-instance, and on the right, we see the corresponding \textsc{SCC}-instance (only one universal node and its associated cliques are fully drawn).
        In both cases, a certificate is marked in the input graph, as well as stated explicitly.
      \else
        Covers are marked in gray.
      \fi
    }%
\end{figure}%
\begin{proof}
        $(\Rightarrow)\colon$
        Let $\mathcal{C}$ be a node clique cover of $G$ with $|\mathcal{C}| \leq s$.
        Without loss of generality, we assume that $\mathcal{C}$ is a partition of $V$---for otherwise if there are distinct $C', C'' \in \mathcal{C}$ with $C' \cap C'' \neq \emptyset$,
        then $\mathcal{C}' \coloneqq (\mathcal{C} \setminus \set{C'}) \cup \set{C' \setminus C''}$ is a node clique cover of $G$
        with $|\mathcal{C}'| = |\mathcal{C}|$ and the number of nodes that are contained in more than one clique strictly less.
        Thus, applying this observation a sufficient number of times always yields a partition of $V$.
        \begin{align*}
            \intertext{Let}
            \mathcal{A} &\coloneqq \left \{ C \cup \{u_i\} \mid C \in \mathcal{C}, 1 \leq i \leq \ell \right \} \text{ and}\\
            \mathcal{B} &\coloneqq \left \{ \{ v_1, v_2 \} \mid v_1v_2 \in E \right \}.
        \end{align*}
        We claim that $\mathcal{A} \cupdot \mathcal{B}$ is a sigma clique cover of $G^{\ell}$ with
        \iflong
          \[
            \wgt(\mathcal{A} \cupdot \mathcal{B}) \leq \ell (|V| + s + 1) - 1.    
          \]
        \else
          $\wgt(\mathcal{A} \cupdot \mathcal{B}) \leq \ell (|V| + s + 1) - 1$.
        \fi
        First, we verify that $\mathcal{A} \cupdot \mathcal{B}$ conforms to \cref{definition:scc}, that is, it indeed is a sigma clique cover of $G^\ell$.
        To that end, we begin by verifying that $G[C]$ is a clique for all $C \in \mathcal{A} \cupdot \mathcal{B}$.
        By construction, we need to differentiate two cases:
        Firstly, let $C \cup \{ u_i \} \in \mathcal{A}$. Since $G[C]$ is a clique, $E \subseteq E(G^{\ell})$ and $\forall v \in V \colon u_iv \in E(G^{\ell})$, it follows that $G^{\ell}[C \cup \{ u_i \}]$ is also a clique.
          Secondly, let $\{v_1, v_2\} \in \mathcal{B}$. Similarly, since $G[\{v_1, v_2\}] \simeq K_2$ and $E \subseteq E(G^{\ell})$, we have $G^{\ell}[\{v_1, v_2\}] \simeq K_2$.

        Now, we prove that all edges of $G^\ell$ are ``covered'' by $\mathcal{A} \cupdot \mathcal{B}$.  Two cases need to be verified:
        Consider any $v_1v_2 \in E$, i.e.,\ those edges that are ``inherited'' from $G$ to $G^{\ell}$. We see that $\{v_1, v_2\} \subseteq B$ by definition.
        Furthermore, consider any $u_iv \in E(G^{\ell}) \setminus E$, i.e.,\ those edges added to $G$ in the construction of $G^{\ell}$. Observe that since $\exists C \in \mathcal{C}$ with $v \in C$, we have $\{u_i, v\} \subseteq C \cup \{u_i\} \in \mathcal{A}$.
        
        We conclude that $\mathcal{A} \cupdot \mathcal{B}$ is a sigma clique cover of $G^\ell$
        and proceed to verify that the claimed bound on the weight holds. By definition of $\mathcal{A}$, we obtain

        \ArrowsC{9.94cm}{
            \wgt(\mathcal{A}) &= \sum_{\substack{C \in \mathcal{C},\\1 \leq i \leq \ell}} |C \cup \{u_i\}| \Arrow{$\{u_1, \ldots, u_{\ell}\} \cap C = \emptyset$ for all $C \in \mathcal{C}$}\\
                           &= \ell \sum_{C \in \mathcal{C}} |C| + 1 \\
                           &= \ell \left ( |\mathcal{C}| + \sum_{C \in \mathcal{C}} |C| \right ) \Arrow{\hbox{$\mathcal{C}$ is a partition of $V$}} \\
                           &= \ell \left ( |\mathcal{C}| + |V| \right ).
        }
        Clearly, $\wgt(\mathcal{B}) = 2 |E| = \ell - 1$.
        Now, using $\mathcal{A} \cap \mathcal{B} = \emptyset$, we derive
        \ArrowsC{6.8cm}{
            \wgt(\mathcal{A} \cupdot \mathcal{B}) &= \wgt(\mathcal{A}) + \wgt(\mathcal{B}) \\
                                               &= \ell ( |\mathcal{C}| + |V| + 1) - 1 \Arrow{$|\mathcal{C}| \leq s$.} \\
                                               &\leq \ell (s + |V| + 1) - 1.
                                             }
        Thus, the forward direction of the proof is established.

        \smallskip
            
        \noindent $(\Leftarrow)\colon$
        Let $\mathcal{S}$ be a sigma clique cover of $G^\ell$ with
          \begin{align*}
            \wgt(\mathcal{S}) &\leq \ell (|V| + s + 1) - 1,
          \end{align*}
        and let
        \begin{align*}%
            u^{*}       &\in \argmin_{u \in \{u_1, \ldots, u_{\ell}\}} \wgt \left ( \{ C \in \mathcal{S} \mid u \in C \} \right ), \\
            \mathcal{X} &\coloneqq \left \{ C \in \mathcal{S} \mid u^{*} \in C \right \} \text{, and} \\
            \mathcal{N} &\coloneqq \left \{ C \setminus \{ u^{*} \} \mid C \in \mathcal{X} \right \}.
        \end{align*}%
        We claim that $\mathcal{N}$ is a node clique cover of $G$ with $|\mathcal{N}| \leq s$.

        First, we verify that $\mathcal{N}$ conforms to \cref{definition:ncc}, i.e.,\ it indeed is a node clique cover of $G$.
        Clearly, $G[C]$ is a clique for all $C \in \mathcal{N}$.
        It remains to verify that all vertices of $G$ are ``covered'' by $\mathcal{N}$:
        Let $v \in V$. Since $\mathcal{S}$ is a sigma clique cover of $G^{\ell}$ and $vu^{*} \in E(G^{\ell})$, there is some $C \in \mathcal{S}$ s.t.\ $\{v, u^{*}\} \subseteq C$.
        It immediately follows that $v \in C \setminus \{u^{*}\} \in \mathcal{N}$.
        
        Second, we establish that $|\mathcal{N}| \leq s$.
        To that end, first, we derive $\wgt(\mathcal{X}) \le |V| + s$. Towards a contradiction, suppose that $\wgt(\mathcal{X}) \ge |V| + s + 1$. Observe that since no $C \in \mathcal{S}$ can contain two different universal nodes of $G^{\ell}$ we get
        \ArrowsC{11.1cm}{
            \wgt(\mathcal{S}) &\ge %
                \sum_{u \in \{u_1, \ldots, u_{\ell}\}} \wgt \left ( \{ C \in \mathcal{S} \mid u \in C \} \right ) \Arrow{choice of $u^*, \mathcal{X}$} \\
                 &\ge \ell \cdot \wgt(\mathcal{X}) \Arrow{$\wgt(\mathcal{X}) \ge |V| + s + 1$} \\
                 &\ge \ell (|V| + s + 1) \\
                 &=   \ell (|V| + s + 1) - 1 + 1 \Arrow{$\ell (|V| + s + 1) - 1 \geq \wgt(\mathcal{S})$}\\
                 &\ge \wgt(\mathcal{S}) + 1.
        }
        In total, this yields $\wgt(\mathcal{S}) \ge \wgt(\mathcal{S}) + 1$, hence $\wgt(\mathcal{X}) \le |V| + s$.

        \smallskip

        \noindent Now, towards the final contradiction, suppose $|\mathcal{N}| \ge s + 1$. We obtain
        \ArrowsC{10.4cm}{
            \wgt(\mathcal{X}) &= \sum_{\substack{C \in \mathcal{S}, \\ u^{*} \! \in C}} |C| \Arrow{definition of $\mathcal{N}$} \\
                 &= \sum_{C \in \mathcal{N}} \left |C \cup \{u^{*}\} \right | \Arrow{$u^* \not\in C$ for all $C \in \mathcal{N}$} \\
                 &= |\mathcal{N}| + \sum_{C \in \mathcal{N}} |C| \Arrow[]{$|\mathcal{N}| \ge s + 1$} \\
                 &\ge s + 1 + \sum_{C \in \mathcal{N}} |C|  \Arrow[]{double counting principle} \\
                 &= s + 1 + \sum_{v \in V} |\{ C \in \mathcal{N} \mid v \in C \}| \Arrow[]{ $\mathcal{N}$ covers $V$ } \\%, because $\mathcal{N}$ is a node clique cover of $G$.} \\
                 &\ge s + 1 + |V|.
        }
        Thus, we have derived both $\wgt(\mathcal{X}) \ge |V| + s + 1$ and $\wgt(\mathcal{X}) \le |V| + s$, a contradiction. Hence, we conclude that $|\mathcal{N}| \leq s$.
\qed
\end{proof}

Using this preliminary work, the \NP-completeness proof is straightforward:
\begin{restatable}[\appsymb]{theorem}{sccnpcompletetheorem}
  \label{theorem:scc_np_complete}
  \textsc{Sigma Clique Cover} is \NP-complete.
\end{restatable}
\appendixproof{theorem:scc_np_complete}{
\ifshort\sccnpcompletetheorem*\fi
  \begin{proof}
    \cref{lemma:scc_ncc_reduction} directly yields a polynomial-time many-one reduction from \textsc{NCC} to \textsc{SCC}, i.e.,\ deciding an instance 
    $(G,s)$ of \textsc{NCC} is equivalent to deciding the instance $\left (G^{\ell},\ell \left ( |V| + s + 1 \right ) - 1 \right )$ of \textsc{SCC}
    where $\ell \coloneqq 2 |E| + 1$. Because \textsc{NCC} is \NP-hard \cite{karp}, so is \textsc{SCC}.
    Observe that $\text{\textsc{SCC}} \in \text{\NP}$, since a certificate for \textsc{SCC} can clearly be guessed and checked in polynomial-time.
    Consequently, we conclude that \textsc{SCC} is \NP-complete.
    \qed
  \end{proof}
}

\section{NP-Completeness of \textsc{Cluster Vertex Splitting}}
\label{section:cvs_np_complete}
\appendixsection{section:cvs_np_complete}
\ifshort
We use \cref{theorem:scc_np_complete} to show that also \pCVS\ is NP-complete:
\begin{problembox}[framed]{\textsc{\problempad Cluster Vertex Splitting (CVS)}}
    \problempad Input: & A tuple $(G, k)$, where $G$ is a graph and $k \in \mathbb{N}$. \\
    \problempad Question: & Is there a sequence of at most $k$ vertex splits that transforms $G$~\vphantom{\problempad}\linebreak into a cluster graph?
\end{problembox}
The basic idea is to show that, in an $n$-vertex graph, finding a sigma clique cover with weight $n + k$ is equivalent to finding $k$ vertices to split such that the resulting graph is a cluster graph:
Given such a sigma clique cover, we may look at each clique and its overlap to the rest of the graph.
We can split each such clique off the rest of the graph by splitting all vertices in the overlap.
This results in splitting each vertex a number of times equal to the number of times it is covered by a clique minus one, that is, overall $k$ vertex splits.
In the other direction, taking a cluster graph obtained by splitting and projecting it onto the original vertices in the input graph will yield a sigma clique cover.
Its weight corresponds to the sum of the sizes of the clusters, which is exactly the number of copies we have created, that is, $n + k$.
See the full version of this paper \cite{DBLP:journals/corr/abs-2309-00504} for a formal proof.

\fi

\toappendix{

We will now build upon the \NP-completeness of \textsc{SCC} and attend
to the \NP-completeness proof of \textsc{CVS}. The formal problem definition of the corresponding decision problem is given below.
\begin{problembox}[framed]{\textsc{\problempad Cluster Vertex Splitting (CVS)}}
    \problempad Input: & A tuple $(G, k)$, where $G$ is a graph and $k \in \mathbb{N}$. \\
    \problempad Question: & Is there a sequence of at most $k$ vertex splits that transforms $G$~\vphantom{\problempad}\linebreak into a cluster graph?
\end{problembox}
The reduction will be accomplished in a multi-step manner:
We begin with introducing two lemmata, \cref{lemma:cvs_scc_reducion_forward_lemma} and \cref{lemma:cvs_scc_reducion_backward_lemma},
used to prove the forward and backward direction of \cref{lemma:cvs_scc_reduction}, respectively.
Then, in \cref{lemma:cvs_scc_reduction}, we establish a close correspondence between instances of \textsc{SCC} and instances of \textsc{CVS}.
Finally, in \cref{theorem:cvs_np_complete}, we use said correspondence to show that \textsc{CVS} is \NP-complete.

\cref{lemma:cvs_scc_reducion_forward_lemma} essentially states the following:
Consider a graph $G'$ that has a sigma clique cover $\mathcal{C}'$.
If we merge two non-adjacent vertices $v$ and $w$ in $G'$ into a vertex we call $u$, that is, we perform a reverse vertex split, 
we obtain a new graph, $G$.
Then, we can replace each occurrence of $v$ or $w$ in $\mathcal{C}'$ with $u$ and obtain a sigma clique cover $\mathcal{C}$ of the same weight for $G$.
Note that the ``overlap'' of $\mathcal{C}$, $\wgt(\mathcal{C}) - \N{G}$, is one more than the ``overlap'' of $\mathcal{C}'$, $\wgt(\mathcal{C'}) - \N{G'}$.

\begin{restatable}{lemma}{cvssccreducionforwardlemma}
  \label{lemma:cvs_scc_reducion_forward_lemma}
  Let $G = (V, E)$ be a graph and let
  $G' = (V', E')$ be obtained from $G$ by splitting $u \in V$ into $v, w \in V'$.
  If $\mathcal{C}'$ is a sigma clique cover of $G'$,
  then there exists a sigma clique cover $\mathcal{C}$ of $G$ with $\wgt(\mathcal{C}) = \wgt(\mathcal{C}')$.
\end{restatable}
  \begin{proof}
    Using
    \begin{align*}
      f(C') \coloneqq
      \begin{cases}
        \parens{C' \setminus \set{v, w}} \cup \set{u} & \text{if } C' \cap \{v, w\} \neq \emptyset \\
        C'                                            & \text{otherwise}
      \end{cases}
    \end{align*}
    we define
    \[
      \mathcal{C} \coloneqq \set{f(C') \mid C' \in \mathcal{C}'}.
    \]
    Note that $f$ gives a bijection from $\mathcal{C'}$ to $\mathcal{C}$.

    \smallskip
    We claim that $\mathcal{C}'$ satisfies the conditions of this lemma.
    First, we establish that $\mathcal{C}$ is a sigma clique cover of $G$ by verifying the two conditions of \cref{definition:scc}.
    We begin by proving that all $C \in \mathcal{C}$ induce cliques in $G$.

    Let $C \in \mathcal{C}$.
    Assume $f^{-1}(C) = C$.
    Observe that $C \cap \set{v, w} = \emptyset$.
    This implies that $G[C] = G'[C]$.
    Hence, since $G'[C]$ is a clique, so is $G[C]$.

    Conversely, assume $f^{-1}(C) \neq C$.
    Without loss of generality, we assume that $v \in f^{-1}(C)$ and $w \not\in f^{-1}(C)$, since $f^{-1}(C)$ cannot both contain $v$ and $w$ by the semantics of vertex splitting.

    Towards the goal of showing $v_1 v_2 \in E$, let $v_1, v_2 \in C$ with $v_1 \neq v_2$.

    \begin{description}
        \item[\emph{Case $\set{v_1, v_2} \cap \set{u} = \emptyset$}:]
        We get that $v_1 v_2 \in E$ if and only if $v_1 v_2 \in E'$ by the way our vertex split was defined.
        From $\set{v_1, v_2} \subseteq f^{-1}(C)$, our assumption that $ f^{-1}(C) \in \mathcal{C}' $ is a sigma clique cover of $G'$ and the correspondence just established, it follows that $v_1 v_2 \in E$.

        \item[\emph{Case $\set{v_1, v_2} \cap \set{u} \neq \emptyset$}:]
        Without loss of generality, assume $v_1 = u$.
        Since $\set{v, v_2}$ is a subset of $f^{-1}(C)$, again invoking that $\mathcal{C}'$ is a sigma clique cover of $G'$ to derive $v v_2 \in E'$ and $N_{G'}(v) \subseteq N_G(u)$, it follows that $u v_2 \in E$.
    \end{description}

    Now, we prove the second property, that is, all edges of $G$ are covered by $\mathcal{C}$.
    Again, let $v_1, v_2 \in C$ with $v_1 \neq v_2$.
    \begin{description}
        \item[\emph{Case $\set{v_1, v_2} \cap \set{u} = \emptyset$}:]
        This edge is not affected by the split, therefore $v_1 v_2 \in E'$, enabling us to choose $C' \in \mathcal{C}'$ such that$\set{v_1, v_2} \subseteq C'$.
        Thus $\set{v_1, v_2} \subseteq C' \setminus \set{v, w} \subseteq f(C') \in \mathcal{C}$.
        \item[\emph{Case $\set{v_1, v_2} \cap \set{u} \neq \emptyset$}:]
        Without loss of generality, assume $v_1 = u$. By the semantics of our split, either $v v_2 \in E'$ or $w v_2 \in E'$ must hold. Without loss of generality, assume the former.
        By the assumption of $\mathcal{C'}$ being a sigma clique cover of $G'$, we can choose $C'$ such that $\set{v, v_2} \subseteq C' \in \mathcal{C}'$.
        Thus, we find that $\set{u, v_2} \subseteq f(C') \in \mathcal{C}$.
    \end{description}
    Therefore, $\mathcal{C}$ is a sigma clique cover of $G$.
    Finally, observe that $f$ ranging over $\mathcal{C'}$ does not change the cardinality of any image it maps, implying that $\wgt(\mathcal{C}) = \wgt(\mathcal{C'})$.
\qed
\end{proof}

Now, we tend to the other direction.
In essence, \cref{lemma:cvs_scc_reducion_backward_lemma} states the following:
Consider a graph $G$ that has a sigma clique cover $\mathcal{C}$ of ``overlap'', that is, $\wgt(\mathcal{C}) - \N{G}$, at most $\alpha \in \mathbb{N}$.
If $\alpha$ is zero, then $G$ evidently is a cluster graph.
Otherwise, there is a vertex $u$ covered by at least two cliques, $C_1$ and $C_2$.
Then, we can define a vertex split acting on $u$ that ``pulls the clique $C_1$ away from the other cliques of $\mathcal{C}$'' while leaving the cliques of the sigma clique cover intact.
One of $u$'s descendants is then only covered by a single clique. Reference \cref{figure:split_sigma_clique_cover} for an illustration.
Consequently, we obtain a graph $G'$ that has a sigma clique cover of the same weight, but with an ``overlap'' decremented by one.

\begin{restatable}{lemma}{cvssccreducionbackwardlemma}%
  \label{lemma:cvs_scc_reducion_backward_lemma}
  Let $G = (V, E)$ be a graph without isolated vertices and let
  $\mathcal{C}$ be a sigma clique cover of $G$ with
  $\wgt(\mathcal{C}) \leq |V| + \alpha \in \mathbb{N}$ as well as
  $|C| > 1$ for all $C \in \mathcal{C}$.
  Then, either $G$ is already a cluster graph or there is $u \in V$
  such that $u$ can be split in $G$ to obtain $G' = (V', E')$ satisfying
  \begin{enumerate}
  \item $G'$ has a sigma clique cover $\mathcal{C}'$,
  \item $\wgt(\mathcal{C}') \leq |V'| + \alpha - 1$,
  \item $|C'| > 1$ for all $C' \in \mathcal{C}'$, and
  \item $G'$ does not contain isolated vertices.
  \end{enumerate}
\end{restatable}
  \newcommand{\uin} {u_{\text{in}}}
  \newcommand{\uout} {u_{\text{out}}}
  \begin{figure}[t]
    \begin{center}
        \includegraphics{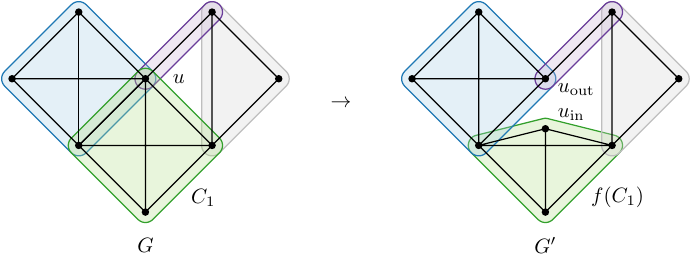}
    \end{center}
    
    \caption[Illustration for the proof of \cref{lemma:cvs_scc_reducion_backward_lemma}.]{
    On the left, a graph $G$ with a sigma clique cover $\mathcal{C}$ is depicted.
    The clique $C_1 \in \mathcal{C}$ is marked in green.
    On the right, a graph $G'$, obtained by splitting $u$ into $u_\text{in}$ and $u_\text{out}$, is drawn.
    Additionally, a sigma clique cover $\mathcal{C}'$ of $G'$ is shown.
    The clique $C_1$ of $\mathcal{C}$ was ``pulled away'' to form $f(C_1)$ in the derived $\mathcal{C}'$, creating a sigma clique cover of ``decreased overlap''.
    }
    \label{figure:split_sigma_clique_cover}
  \end{figure}
  \begin{proof}

    If $G$ is not already a cluster graph, there must exist $C_1 \neq C_2 \in \mathcal{C}$ such that $C_1 \cap C_2 \neq \emptyset$. In this case, let $u \in C_1 \cap C_2$.

    We define $G' = (V', E')$ as the graph that is obtained when $u$ is split into the two vertices $\uin$ and $\uout$ obeying:
    \begin{gather*}
        N_{G'}(\uin) \coloneqq N_G(u) \cap C_1, \\
        N_{G'}(\uout) \coloneqq \left ( N_G(u) \setminus C_1 \right ) \cup \left \{ v \in N_G(u) \cap C_1 \mid \exists C \in \mathcal{C} \setminus \{C_1\} \colon u, v \in C \right \}. 
    \end{gather*}

    \noindent Furthermore, using the map
    \begin{align*}
        f(C) \coloneqq
        \begin{cases}
            \parens{C \setminus \set{u}} \cup \set{\uin}  & \text{if } C = C_1 \\
            \parens{C \setminus \set{u}} \cup \set{\uout} & \text{if } u \in C \land C \neq C_1 \\
            C & \text{otherwise}
        \end{cases}
    \end{align*}
    we can define
    \[
        \mathcal{C}' \coloneqq \set{ f(C) \mid C \in \mathcal{C}}.
    \]

    Note that $f$ gives a bijection between $\mathcal{C}$ and $\mathcal{C}'$; thus $f^{-1}(\cdot)$ will be used to denote a single well-defined element in what follows.

    Intuitively, this split corresponds to ``pulling out'' the vertex $u$ creating $\uin$,
    only keeping the part of $u$'s neighborhood contained in $C_1$,
    so that $\uin$ will only be contained in a single clique $f(C_1)$ in the derived sigma clique cover $\mathcal{C}'$
    and letting $\uout$ inherit the rest of the neighborhood, plus a select set of vertices already neighbors of $\uin$,
    as to not destroy any cliques of $\mathcal{C} \setminus \set{C_1}$.
    See \cref{figure:split_sigma_clique_cover} for an example.

    \smallskip
    We claim that $G'$ and $C'$ satisfy Condition (1)--(4) of this lemma. Observe that $f$ preserves the cardinality of mapped sets, and that $|C| > 1$ for all $C \in \mathcal{C}$. Thus, Condition 3 follows immediately.
    It remains to show that Conditions 1, 2 and 4 are satisfied.

        \subparagraphx{Condition 1.} To establish that $\mathcal{C}'$ is a sigma clique cover of $G'$, the two conditions of \cref{definition:scc} need to be verified.
        We start with the first condition, that is, we verify that all $C' \in \mathcal{C}'$ induce cliques in $G'$:

            Let $C' \in \mathcal{C}'$.
            Since $|C'| < 2$ is impossible, we select arbitrary $v_1, v_2 \in C'$ such that $v_1 \neq v_2$.
            We denote the intersection of $\set{\uin, \uout}$ and $\set{v_1, v_2}$ by $I$ and enumerate all arising cases:
            \begin{description} %
                \item[$I = \emptyset$:] We have $v_1v_2 \in E$ since $ \set{v_1, v_2} \subseteq f^{-1}(C')$, further implying $v_1v_2 \in E'$, because this edge was not affected by the splitting operation.
                
                \item[$I = \set{\uin}$:] Without loss of generality, assume $v_1 = \uin$.
                Since $f^{-1}(C') = C$ and $\set{v_2, u} \subseteq C$, we obtain $v_2 \in N_G(u)$ and $v_2 \in C$. Thus, $v_2 \in N_{G'}(\uin)$, implying $v_1v_2 \in E'$ by construction.
                
                \item[$I = \set{\uout}$:]
                Without loss of generality, assume $v_1 = \uout$. Observe that this yields $v_2 \in N_G(u)$.
                In the case that $v_2 \not\in C_1$ it holds that $v_2 \in N_G(u) \setminus C_1 \subseteq N_{G'}(\uout)$, thus $\uout v_2 = v_1v_2 \in E'$.

                Conversely, if $v_2 \in C_1$, observe that $\set{u, v_2} \in f^{-1}(C')$.
                This implies $v_2 \in N_G(u)$,
                and using our assumption, we get $v_2 \in N_G(u) \cap C_1$.
                Note also that $f^{-1}(C') \neq C_1$ by definition of $f$.

                Thus, $f^{-1}(C')$ serves as a witness for \\ $ v_2 \in \left \{ v \in N_G(u) \cap C_1 \mid \exists C' \in \mathcal{C} \setminus C_1 \colon u, v \in C' \right \} \subseteq N_{G'}(\uout)$.
                Hence, we have $\uout v_2 = v_1v_2 \in E'$.
                
                \item[$I = \set{\uin, \uout}$:] Contradiction to the definition of the split yielding $G'$.
            \end{description}

            \smallskip Now, we proceed with the second condition, demanding that all edges of $G'$ be covered by $\mathcal{C}'$:
            Let $v_1v_2 \in E'$. Again, we denote the intersection of $\set{\uin, \uout}$ and $\set{v_1, v_2}$ by $I$ and enumerate all arising cases:
            \begin{description}
                \item[$I = \emptyset$:] Since this case mandates that $v_1v_2 \in E$, by assumption of $\mathcal{C}$ being a sigma clique cover of $G$, there exists $C \in \mathcal{C}$ such that $\set{v_1, v_2} \subseteq{C}$.
                By definition of $f$, it must also hold that $\set{v_1, v_2} \subseteq{f(C)}$. 
                \item[$I = \set{\uin}$:]
                Without loss of generality, assume $v_1 = \uin$. Because $v_2 \in N_{G'}(\uin) \subseteq C_1$, and $v_2 \neq u$, we know that $v_2 \in f(C_1)$. Furthermore, because $\uin \in f(C_1)$ by definition of $f$, we have $\uin v_2 \in E(G'[f(C_1)])$.
                \item[$I = \set{\uout}$:]
                Without loss of generality, assume $v_1 = \uout$. It holds that $v_2 \in N_{G'}(\uout)$. As $N_{G'}(\uin)$ is defined as the union of two sets, we distinguish two cases:
                Firstly, assume $v_2 \in N_G(u) \setminus C_1$. By definition of the vertex split at hand, we have $u v_2 \in E$. Using the assumption that $\mathcal{C}$ is a sigma clique cover of $G$, there is $C^* \in \mathcal{C} \setminus \{C_1\}$ with $\{u, v_2\} \subseteq C^*$.
                By the second case of the definition of $f$, it thus follows that $\{\uout, v_2\} \subseteq f(C^*)$.

                Secondly, assume $v_2 \in \left \{ v \in N_G(u) \cap C_1 \mid \exists C \in \mathcal{C} \setminus \{C_1\} \colon u, v \in C \right \}$.
                This yields that there is $C^* \in \mathcal{C} \setminus \{C_1\}$ with $\{u, v_2\} \subseteq C^*$ and therefore, by the argument employed in the previous case, we have $\{\uout, v_2\} \subseteq f(C^*)$.
                    
                \item[$I = \set{\uin, \uout}$:]
                Contradiction to the definition of the split yielding $G'$.
            \end{description}
            
        Thus, $\mathcal{C}'$ indeed is a sigma clique cover of $G'$.

        \subparagraphx{Condition 2.} Since $|C| = |f(C)|$ for all $C \in \mathcal{C}$, we have
        \ArrowsN{0}{
              \wgt(\mathcal{C}') &= \wgt(\mathcal{C})  \Arrow{by assumption} \\
                                &\leq |V| + \alpha \Arrow{$|V| = |V'| - 1$.} \\
                                &\leq |V'| + \alpha - 1.
        }
        
        \subparagraphx{Condition 4.} Towards a contradiction, suppose $G'$ contains an isolated vertex $v \in V'$.
        As the vertex degree of all vertices, except those of $\uin$ and $\uout$, are necessarily inherited from $G$ by the vertex split,
        we must have either $N_{G'}(\uin) = \emptyset$ or $N_{G'}(\uout) = \emptyset$.

        Suppose $N_{G'}(\uin) = \emptyset$.
        Since $u \in C_1$ and $|C_1| > 1$, there exists $v_2 \neq u$ with $v_2 \in C_1$. Since $G[C_1]$ is a clique, we get $v_2 \in N_G(u)$. Therefore, $v_2 \in N_G(u) \cap C_1 = N_{G'}(\uin)$, contradicting $N_{G'}(\uin) = \emptyset$.

        Now, suppose $N_{G'}(\uout) = \emptyset$.
        Invoking the same argument as in the last case substituting $C_2$ for $C_1$, we derive $v_2 \in N_G(u)$.
        First, suppose $v_2 \in C_1$.
        Using $C_2$ as witness, we obtain
            \begin{equation*}
                v_2 \in \left \{ v' \in N_G(u) \cap C_1 \mid \exists C \in \mathcal{C} \setminus \{C_1\} \colon u, v' \in C \right \} \subseteq N_{G'}(\uout),
            \end{equation*}
             which contradicts $N_{G'}(\uout) = \emptyset$.

        Now, suppose the contrary, that is, $v_2 \not\in C_1$. 
        We derive 
            \begin{equation*}
                v_2 \in N_G(u) \setminus C_1 \subseteq N_{G'}(\uout),
            \end{equation*}
            which again is a contradiction to $N_{G'}(\uout) = \emptyset$.
            
        Thus, our initial assumption that $G'$ contains an isolated vertex $v \in V'$ is invalid.
\qed
\end{proof}

With this groundwork, we can formulate and prove \cref{lemma:cvs_scc_reduction}.
In essence, the lemma states that it is equivalent to search for sigma clique covers of bounded ``overlap'',
and splitting sequences ending in cluster graphs of bounded length.
Note that some special care needs to be taken to deal with the possibility of isolated vertices.

To prove the correspondence, we proceed as follows:
Suppose we are given a graph with a sigma clique cover of ``overlap'' at most $k$.
Then, we can apply \cref{lemma:cvs_scc_reducion_backward_lemma} at most $k$ times to obtain a graph admitting a sigma clique cover of zero ``overlap'', which is a cluster graph.

Conversely, consider a splitting sequence of length at most $k$ that ends in a cluster graph.
The last graph trivially has a sigma clique cover of zero ``overlap''.
Then, we can work through the sequence in reverse order,
and by repeatedly applying \cref{lemma:cvs_scc_reducion_forward_lemma}, obtain a sigma clique cover of the first graph that has an ``overlap'' of at most $k$.

\begin{restatable}{lemma}{cvssccreduction}
  \label{lemma:cvs_scc_reduction}
  Let $G = (V,E)$ be a graph, and let $I \coloneqq \{v \in V \mid d_G(v) = 0\}$. Then, $(G,k)$ is a positive instance of \textsc{CVS} if and only if $(G, |V| - |I| + k)$ is a positive instance of \textsc{SCC}.
\end{restatable}
  \begin{proof}
    $(\Rightarrow)\colon$
    Let $G_0, \ldots, G_\ell$ be a sequence of graphs with $G_0 = G$ and $\ell \leq k$
    such that each graph, except $G_0$, is obtained from its predecessor via a vertex split,
    and $G_\ell$ is a cluster graph.
    Observe that a vertex split never results in a graph with fewer isolated vertices than the original graph,
    hence at least $|I|$ vertices of $G_\ell$ are isolated.
    By identifying all connected components of $G_\ell$ with their vertex sets, but omitting some $|I|$ trivial components,
    we can construct a sigma clique cover $\mathcal{C}_\ell$ of $G_\ell$ with $\wgt(\mathcal{C}_\ell) = |V(G_\ell)| - |I|$.
    Each split used in the construction of $G_0, \ldots, G_\ell$ introduces exactly one new vertex,
    therefore $|V(G_\ell)| = |V| + \ell$. Combining this with the fact that $\ell \leq k$,
    we derive $\wgt(\mathcal{C}_\ell) \leq |V| - |I| + k$.
    Using the sequence $G_0, \ldots, G_\ell$ in reverse order, we iteratively apply \cref{lemma:cvs_scc_reducion_forward_lemma} $\ell$ times
    using $\mathcal{C}_\ell$ and $G_\ell$ as base case and obtain $\mathcal{C}_0, \ldots, \mathcal{C}_\ell$.
    In particular, it follows that $\mathcal{C}_0$ is a sigma clique cover of $G$ satisfying $\wgt(\mathcal{C}_0) \leq |V| - |I| + k$.
    Thus, $(G, |V| - |I| + k)$ is a positive instance of \textsc{SCC}.

    \smallskip

    $(\Leftarrow)\colon$ Let $\mathcal{C}$ be a sigma clique cover of $G$ with $\wgt(\mathcal{C}) \leq |V| - |I| + k$.
    Without loss of generality, we can assume that $\mathcal{C}$ contains no $C \in \mathcal{C}$ with $|C| \leq 1$,
    for $\mathcal{C} \setminus \set{C}$ still is a sigma clique cover of $G$ of weight not exceeding that of $\mathcal{C}$ for any such $C \in \mathcal{C}$.
    Observe that $\mathcal{C}$ is a sigma clique cover of $H_0 \coloneqq G[V \setminus I]$ too, since $E(H_0) = E(G)$. Furthermore, set $\mathcal{C}_0 \coloneqq \mathcal{C}$.
    By iteratively applying \cref{lemma:cvs_scc_reducion_backward_lemma} for a number of times, call it $\ell$,
    either until a cluster graph is obtained as a direct result of the lemma, or alternatively, stopping after $l=k$ iterations,
    we can obtain the sequences $H_0, \ldots, H_{\ell}$ and $\mathcal{C}_0, \ldots, \mathcal{C}_{\ell}$.

    We shall now verify that also in the latter case where $\ell = k$, $H_\ell$ must be a cluster graph. As a consequence of the $k$ applications of \cref{lemma:cvs_scc_reducion_backward_lemma}, we get $\wgt(H_\ell) \leq |V(H_\ell)|$.
    By considering the fact that for each vertex $v \in V(H_\ell)$ there exists $C \in \mathcal{C}_\ell$ with $v \in C$ (since $\mathcal{C}_\ell$ is a sigma clique cover of $H_\ell$ and $H_\ell$ contains no isolated vertices), we derive $\wgt(H_\ell) \ge |V(H_\ell)|$.
    Thus, we have that $\wgt(H_\ell) = |V(H_\ell)|$ and it follows that $\mathcal{C}_\ell$ forms a partition of $V(H_\ell)$.
    Using this partition property and the fact that $\mathcal{C}_\ell$ is a sigma clique cover of $H_\ell$ allows us to directly conclude that $H_{\ell}$ is a cluster graph.
    Thus, $H_{\ell}$ is a cluster graph in both cases.

    We reintroduce the isolated vertices $I$ by constructing
    \[
        H'_0, \ldots, H'_{\ell} \coloneqq \left ( \vphantom{\sum} \left ( V(H_i) \cupdot I, E(H_i) \right ) \right )_{i \in \set{0, \dots, \ell}}.
    \]
    $H'_0, \ldots, H'_{\ell}$ forms a sequence of graphs where each constituent except the first is generated by performing a split in its predecessor for a total of no more than $k$ splits;
    this property is inherited from $H_0, \ldots, H_{\ell}$.
    Note that in particular $H'_0 = G$ by definition, and furthermore, $H'_{\ell}$ is a cluster graph,
    since adding isolated vertices to a cluster graph yields another cluster graph.
    In total, we thus have obtained a certificate $H'_0, \ldots, H'_{\ell}$ proving that $(G, k)$ is a positive instance of \textsc{CVS}.
    \qed
\end{proof}

With the correspondence just established, the \NP-hardness proof of \textsc{CVS} becomes immediate.
}                               %

\begin{restatable}[\appsymb]{theorem}{cvsnpcompletetheorem}
  \label{theorem:cvs_np_complete}
  \textsc{Cluster Vertex Splitting} is \NP-complete.
\end{restatable}
  \appendixproof{theorem:cvs_np_complete}{
\ifshort\cvsnpcompletetheorem*\fi
\begin{proof}
    Let $(G, s)$ be an instance of \textsc{SCC} and $I \coloneqq \set{v \in V(G) \mid d_G(v) = 0}$.
    We can leverage \cref{lemma:cvs_scc_reduction} to conclude that deciding this instance of \textsc{SCC} is
    equivalent to deciding the instance $(G, s - |V(G)| + |I|)$ of \textsc{CVS}.
    We have thus constructed a polynomial-time many-one reduction from \textsc{SCC} to \textsc{CVS}.
    Because \textsc{SCC} is \NP-hard by \cref{theorem:scc_np_complete}, so is \textsc{CVS}.
    Observe that $\text{\textsc{CVS}} \in \text{\NP}$, since a certificate for \textsc{CVS} can clearly be guessed and checked in polynomial-time.
    Consequently, we conclude that \textsc{CVS} is \NP-complete.
\qed
\end{proof}
}                               %

\section{A Linear Kernel for \textsc{Cluster Vertex Splitting}}
\label{chapter:cluster_kernel}
\iflong
To start, we introduce the concept of \emph{valency},
a straightforward tool that assists us in counting arguments.
We also review the concept of \emph{critical cliques} \cite{cc},
where vertices that share identical closed neighborhoods are grouped together.
\else
  We first introduce some basic notions in \cref{section:CC}.
\fi
In \cref{section:ruleI}, we establish the groundwork for the first data-reduction rule of the kernel,
which allows us to reduce certain critical cliques in a \textsc{Sigma Clique Cover} instance.
The second rule of the kernel is based on \cref{section:ruleII},
where we determine that \textsc{Sigma Clique Cover} instances that have been exhaustively reduced using the previously explored mechanism and still contain more than $3k$ vertices are negative instances.
We then give the kernel in \cref{section:kernel}.

\subsection{The Notions of Valency and Critical Cliques}
\label{section:CC}

We will frequently have to prove lower bounds for the weight that a sigma clique cover needs to have at minimum.
This we will do by observing that particular vertices must be covered by at least a certain number of cliques each.
\iflong
  To aid in such arguments, we introduce a new measure. The
\else
  For this,
\fi
\emph{valency} of a vertex $v$ with respect to a sigma clique cover $\mathcal{C}$ counts the number of cliques that contain $v$%
\iflong:
\begin{definition}\label{definition:valency}
    Let $\mathcal{C}$ be a sigma clique cover of a graph $G$.
    Then, for each vertex $v \in V(G)$, we define the \emph{valency} of $v$ with respect to $\mathcal{C}$ as the number 
    of cliques in $\mathcal{C}$ that cover $v$. Symbolically, 
    \iflong
      we express this quantity as
    \begin{equation*}
        \val_\mathcal{C}(v) \coloneqq |\set{C \in \mathcal{C} \mid v \in C}|.
    \end{equation*}
  \else
    $\val_\mathcal{C}(v) \coloneqq |\set{C \in \mathcal{C} \mid v \in C}|$.
  \fi
\end{definition}%
\else, that is, $\val_\mathcal{C}(v) \coloneqq |\set{C \in \mathcal{C} \mid v \in C}|$.
\fi

With this notation, we can express the weight of a sigma clique cover in an alternative manner:
Via the definition of $\wgt(\cdot)$ (\cref{definition:scc}) and the principle of double counting, we obtain
\iflong
  \begin{equation*}
    \wgt(\mathcal{C}) = \sum_{C \in \mathcal{C}} |C| = \sum_{v \in V(G)} \val_{\mathcal{C}}(v).
  \end{equation*}
\else
  $\wgt(\mathcal{C}) = \sum_{C \in \mathcal{C}} |C| = \sum_{v \in V(G)} \val_{\mathcal{C}}(v)$.
\fi

Another key tool that we will use in this section is the concept of critical cliques, coined by Lin et al.~\cite{cc}.
The \emph{closed neighborhood} of a vertex $v$ in a graph~$G$ is $N_G(v) \cup \set{v}$.
This allows us to consider an equivalence relation, where vertices of a graph are in the same class if and only if their closed neighborhoods coincide.
The equivalence classes under this relation are called the \emph{critical cliques} of~$G$.
Consider a critical clique $C$ of $G$.
Observe that it is fully connected ``internally'', that is, $G[C]$ is a clique, and that $N_G(v) \setminus C = N_G(w) \setminus C$ for any $v, w \in C$, which means that the vertices of $C$ share a common ``external neighborhood''.

\toappendix{
\begin{figure}[t]
    \begin{center}
        \includegraphics{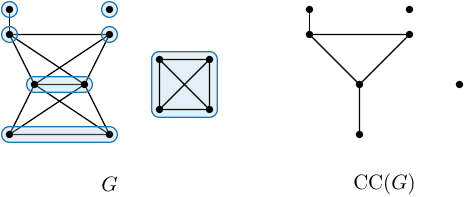}
    \end{center}
    \caption[Example for the concept of critical cliques.]{A graph $G$ whose critical cliques are marked in blue (left) and $\CC(G)$ (right).}
    \label{figure:cc_example}
\end{figure}
}

If we delete all but one vertex from each critical clique, we obtain a graph isomorphic to what we will call
the critical clique graph of $G$; we will use the shorthand $\CC(G)$ to refer to it.
\iflong See \cref{figure:cc_example} for an example.\fi
\iflong
  Formally, we define this graph as follows:
\else
  Formally:
\fi
\begin{definition}
    \label{definition:cc}
    Let $G$ be a graph.
    Consider the equivalence relation $R_G \subseteq V(G) \times V(G)$ where 
    $(v, w) \in R_G$ if and only if $N(v) \cup \set{v} = N(w) \cup \set{w}$.
    We use $[v]_G$ to denote the equivalence class generated by $v \in V(G)$ and $R_G$.
    The \emph{critical clique graph} of $G$, referred to using $\CC(G)$, is given by
    \begin{align*}
        V(\CC(G)) &\coloneqq \set{[v]_G \mid v \in V(G)} \text{ and}\\
        E(\CC(G)) &\coloneqq \set{ [v]_G[w]_G \mid vw \in E(G) \land [v]_G \neq [w]_G}.
    \end{align*}
\end{definition}

The main intuition we make use of here is that members of the same critical clique are essentially ``clones'' of one another.
Thus, it seems reasonable that, provided certain conditions are met, we are allowed to ``shrink'' certain critical cliques
without removing a significant amount of ``computational complexity'' when solving the combinatorial problems we are interested in.

\subsection{Towards a Rule to Shrink Critical Cliques}
\label{section:ruleI}

Consider the critical clique graph $\CC(G)$ of a graph $G$.
We distinguish between two kinds of critical cliques:
\begin{enumerate}
\item Critical cliques $[v]_G$ such that their neighborhood, that is, $N_{\CC(G)}([v]_G)$, forms a clique in $\CC(G)$, and
\item critical cliques $[v]_G$, where said neighborhood does not form a clique.
\end{enumerate}

In this section, we show that, with respect to the sigma clique cover problem, critical cliques of the first kind consisting of at least two vertices, can either safely be reduced in size, or deleted altogether (\cref{lemma:shrink_green}).
Correspondingly, we will refer to them as \emph{reducible critical cliques}.
The second kind of critical cliques we will call \emph{irreducible critical cliques}.

To help prove \cref{lemma:shrink_green}, we first observe that in any minimum-weight sigma clique cover of a graph,
a vertex member of a critical clique of the first kind is always covered by precisely one clique. Furthermore, this clique can be determined explicitly (\cref{lemma:green}).
We start with a useful observation%
\iflong that we prove for completeness' sake\fi:
\begin{restatable}[\appsymb]{lemma}{xlemma}\label{lemma:x}
  Let $G$ be a graph, $\mathcal{C}$ a sigma clique cover of $G$, and
    $v \in C \in \mathcal{C}$. Then, $C \subseteq N_G(v) \cup \set{v}$.
\end{restatable}  
  \appendixproof{lemma:x}{
\ifshort\xlemma*\fi
\begin{proof}
    Suppose there is $w \in C \setminus \set{N_G(v) \cup \set{v}}$.
    Observe that $w$ differs from $v$. But then $w$ cannot be a neighbor of $v$ in $G$.
    Hence, $C$ cannot cover $v$ and $w$ simultaneously, contradicting our choice of $w$.
\qed
\end{proof}
}                               %

Now, we are ready to prove our auxiliary lemma that offers insight into the structure of minimum-weight sigma clique covers:
\begin{restatable}[\appsymb]{lemma}{greenlemmalemma}\label{lemma:green}
    Let $G$ be a graph without isolated vertices
    and let $[v]_G$ be a critical clique in $G$
    such that $\CC(G)[N_{\CC(G)}([v]_G)]$ is a clique.
    Furthermore, let $\mathcal{C}$ be a minimum-weight sigma clique cover of $G$ and let
    $C^* \coloneqq N_G(v) \cup \set{v}$.
    Then, $C^*$ is contained in $\mathcal{C}$. Moreover, $C^*$ is the only clique of $\mathcal{C}$ that covers $v$.
\end{restatable}
\iflong
\begin{proof}
    We will first show that $G[C^*]$ is a clique; this will become useful later on.
    Since $v$ is not isolated, we can select two distinct vertices $a, b \in C^*$. We need to show that $ab \in E(G)$.
    \begin{description}
        \item[Case {$[a]_G = [b]_G = [v]_G \colon$} ] 
        The vertices $a$ and $b$ are part of a shared critical clique.
        Hence, $N_G(a) \cup \set{a} = N_G(b) \cup \set{b}$, which implies $a \in N_G(b)$.

        \item[Case {$[a]_G \neq [v]_G \land [b]_G \neq [v]_G \colon$} ]
        Since $a \neq v$ and $b \neq v$, we have $\set{a, b} \subseteq N_G(v)$,
        implying $\set{va, vb} \subseteq E(G)$. Using \cref{definition:cc}, we obtain that all of 
        $\set{[v]_G[a]_G, [v]_G[b]_G}$ are edges of $\CC(G)$.
        If $[a]_G = [b]_G$, it is immediate that $ab \in E(G)$.
        Otherwise, we invoke the precondition that $\CC(G)[N_{\CC(G)}([v]_G)]$ is a clique,
        yielding $[a]_G[b]_G \in E(\CC(G))$, which implies $ab \in E(G)$.

        \item[Case {$[a]_G = [v]_G \land [b]_G \neq [v]_G \colon$} ]
        Similarly to the last case, $b \neq v$ gives $b \in N_G(v)$, implying $[v]_G[b]_G = [a]_G[b]_G \in E(\CC(G))$.
        Hence, $ab \in E(G)$.

        \item[Case {$[a]_G \neq [v]_G \land [b]_G = [v]_G \colon$} ]
        Symmetrical to the previous case.
    \end{description}

    Next, we show that $v$ is covered by at most one clique.
    Towards a contradiction, suppose that $\val_{\mathcal{C}}(v) \ge 2$.
    Let $C_1$ and $C_2$ be two distinct cliques of $\mathcal{C}$ such that $v \in C_1 \cap C_2$.
    By \cref{lemma:x}, we know that $C_1 \subseteq N_G(v) \cup \set{v}$ and $C_2 \subseteq N_G(v) \cup \set{v}$.
    Thus, $C_1 \cup C_2 \subseteq N_G(v) \cup \set{v} = C^*$. 
    We have already shown that $G[C^*]$ is a clique. Since the family of clique graphs is closed under vertex deletion,
    we thus find that $G[C_1 \cup C_2]$ is a clique too.
    Now, let
    \iflong
      \begin{equation*}
        \mathcal{C'} \coloneqq \parens{\mathcal{C} \setminus \set{C_1, C_2}} \cup \parens{C_1 \cup C_2}.
      \end{equation*}
    \else
      $\mathcal{C'} \coloneqq \parens{\mathcal{C} \setminus \set{C_1, C_2}} \cup \parens{C_1 \cup C_2}$.
    \fi
    Clearly, $\mathcal{C'}$ covers $G$ as $\mathcal{C}$ does. Also, we have just observed that $G[C_1 \cup C_2]$ is a clique, while all other $C \in \mathcal{C'}$ induce cliques in $G$ because $\mathcal{C}$ is a sigma clique cover of $G$.
    Therefore, $\mathcal{C'}$ is a sigma clique cover of $G$.
    But notice
    \ArrowsC{10.36cm}{
        \wgt(\mathcal{C'}) &= \wgt(\mathcal{C}) - |C_1| - |C_2| + |C_1 \cup C_2| \Arrow{$|C_1 \cup C_2| < |C_1| + |C_2|$} \\
                           &< \wgt(\mathcal{C}).
    }
    This contradicts that $\mathcal{C}$ has minimum weight for $G$.
    Therefore, $\val_\mathcal{C}(v) < 2$.
    Since $v$ is not isolated, we additionally have that $\val_\mathcal{C}(v) \ge 1$. Thus, $\val_\mathcal{C}(v) = 1$.
 
    We have shown that $v$ is covered by precisely one clique of $\mathcal{C}$; call it $C$.
    The last remaining step is to prove that $C = C^*$.
    Consider any edge $e \in E(G)$ incident with $v$.
    We observe that $e$ is covered by $C$, for were $e$ covered by any different $C' \in \mathcal{C}$, 
    we would obtain $\val_\mathcal{C}(v) \ge 2$.
    Thus, considering all such edges lets us conclude that $N_G(v) \cup \set{v} = C^* \subseteq C$.
    At the same time, by \cref{lemma:x}, we get $C \subseteq N_G(v) \cup \set{v} = C^*$.
    Therefore, $C$ equals $C^*$ and the proof is complete.
\qed
\end{proof}
\fi

\iflong
It remains to turn our previous observation into a lemma suitable to show the correctness of a reduction rule used in the kernel.
More specifically, when we prove the correctness of Rule~I formulated in \cref{theorem:cluster_kernel}, we will make direct use of the following lemma:
\else
  The next lemma shows the correctness of reducing reducible critical cliques:%
\fi%
\begin{restatable}[\appsymb]{lemma}{shrinkgreenlemma}\label{lemma:shrink_green}
    Let $G$ be a graph without isolated vertices and
    let $[v]_G$ be one of its critical cliques such that $|[v]_G| \ge 2$
    and $\CC(G)[N_{\CC(G)}([v]_G)]$ is a clique.
    Then, $(G, \N{G} + k)$ is a positive instance of \textsc{SCC} if{}f 
    $( G - v, \N{G - v} + k)$ is.
\end{restatable}
\appendixproof{lemma:shrink_green}{
\ifshort\shrinkgreenlemma*\fi
\begin{proof}
    $(\Rightarrow)\colon$
    Let $\mathcal{C}$ be a minimum-weight sigma clique cover of $G$ with $\wgt(\mathcal{C}) \leq \N{G} + k$.
    We apply \cref{lemma:green} and conclude that $v$, as well as all edges incident with $v$, are covered only by a single clique $C^* \in \mathcal{C}$;
    let
    \iflong
    \begin{equation*}
        \mathcal{C'} \coloneqq \parens{\mathcal{C} \setminus C^*} \cup \parens{C^* \setminus \set{v}}.
      \end{equation*}
    \else
      $\mathcal{C'} \coloneqq \parens{\mathcal{C} \setminus C^*} \cup \parens{C^* \setminus \set{v}}$.
    \fi
    Consider $C \in \mathcal{C'}$. We observe that $(G - v)[C] = G[C] - v$ and conclude that $(G - v)[C]$ is a clique, since $G[C]$ is a clique and the class of complete graphs is closed under vertex deletion.
    Furthermore, it is easy to see that since $\mathcal{C}$ covers $G$, we know that $\mathcal{C'}$ covers $G - v$.
    Note also that since $\val_\mathcal{C}(v) = 1$ with $v \in C^* \in \mathcal{C}$, we can deduce $\mathcal{C'} \subseteq \mathcal{P}(V(G - v))$ and $\wgt(\mathcal{C'}) = \wgt(\mathcal{C}) - 1 \leq \N{G} + k - 1= \N{G - v} + k$.
    Thus, $\mathcal{C'}$ is a sigma clique cover of $G - v$ of the required weight.

    $(\Leftarrow)\colon$
    Let $\mathcal{C'}$ be a minimum-weight sigma clique cover of $G - v$ such that $\wgt(\mathcal{C'}) \leq \N{G - v} + k$.
    Furthermore, let $w \in [v]_G \setminus \set{v}$.
    We apply \cref{lemma:green} to $G - v$, $[w]_{G-v}$, and $\mathcal{C'}$ to deduce that there is a single clique $C^* = N_{G - v}(w) \cup \set{w} \in \mathcal{C'}$ where $w \in C^*$.
    Next, let
    \iflong
      \begin{equation*}
        \mathcal{C} \coloneqq \parens{\mathcal{C'} \setminus C^*} \cup \parens{C^* \cup \set{v}}.
      \end{equation*}
    \else
      $\mathcal{C} \coloneqq \parens{\mathcal{C'} \setminus C^*} \cup \parens{C^* \cup \set{v}}$.
    \fi
    We know that $v$ and $w$ are part of the same critical clique in $G$.
    Thus, $N_G(v) \cup \set{v} = N_G(w) \cup \set{w}$.
    Subtracting $v$ on both sides, we obtain
    \iflong
      \begin{equation*}
        N_G(v) = N_{G - v}(w) \cup \set{w} = C^* \subseteq C^* \cup \set{v}.
      \end{equation*}
    \else
      $N_G(v) = N_{G - v}(w) \cup \set{w} = C^* \subseteq C^* \cup \set{v}$.
    \fi
    Thus, all $e \in E(G) \setminus E(G - v)$ are covered by $C^* \cup \set{v}$.
    All remaining edges of $G$ are not incident with $v$; let $e$ be such an edge.
    Since there is $C' \in \mathcal{C'}$ that covers $e$ and $C' \subseteq C$ for some $C \in \mathcal{C}$, 
    we have that $\mathcal{C}$ covers $e$.

    It remains to show that all $C \in \mathcal{C}$ induce cliques in $G$.
    Let $C \in \mathcal{C}$. If $v \not\in C$, then $G[C] = (G - v)[C]$.
    Otherwise, $C$ is equal to $C^* \cup \set{v}$.
    We know that $(G - v)[C^*]$ is a clique and that $G - v \prec G$.
    Thus, we only need to show that all edges between $C^*$ and $\set{v}$ exist in $G$.
    Let
    \iflong
    \begin{align*}
        a \in C^* &= N_{G - v}(w) \cup \set{w}\\    
                  &\subseteq N_G(w) \cup \set{w}\\
                  &= N_G(v) \cup \set{v}.
    \end{align*}
  \else
    \[
      a \in C^* = N_{G - v}(w) \cup \set{w} \subseteq N_G(w) \cup \set{w} = N_G(v) \cup \set{v}.
      \]\fi
      Since $a \neq v$, we have $a \in N_G(v)$, or phrased differently: $av \in E(G)$.

    Observe that $\mathcal{C} \subseteq \mathcal{P}(V(G))$ and $\wgt(\mathcal{C}) = \wgt(\mathcal{C'}) + 1 \leq \N{G - v} + 1 + k = \N{G} + k$.
    Therefore, we can finish our proof and conclude that $\mathcal{C}$ is a sigma clique cover of $G$ of the required weight.
\qed
\end{proof}
}

\subsection{Towards a Rule to Recognize Negative Instances}
\label{section:ruleII}

In the previous section, we laid the foundation for a rule that minimizes the sizes of reducible critical cliques.
Consider an instance $(G, \N{G} + k)$ of \textsc{Sigma Clique Cover} that has been exhaustively reduced using the aforementioned rule.
We now observe that, if this instance has more than $3k$ vertices, then it is a negative instance.
This will serve as the basis for Rule~II defined in \cref{theorem:cluster_kernel}.

We proceed as follows:
We assume that $G$ has more than $3k$ vertices and 
consider an arbitrary sigma clique cover $\mathcal{C}$ of $G$.
Then, we provide two separate lower bounds on $\wgt(\mathcal{C})$.
One bound is based on reducible critical cliques, while the other bound is based on irreducible critical cliques.
Each lower bound individually is too weak, but the maximum of both will be greater than $\N{G} + k$ in all cases, yielding that $(G, \N{G} + k)$ is a negative instance.

\begin{restatable}[\appsymb]{lemma}{threeknegativelemma}\label{lemma:three_k_negative}
    Let $G$ be a graph such that none of its connected components are cliques and $k \in \mathbb{N}$.
    We divide $V(\CC(G))$ into the partition $A \cupdot B$ where $v \in A$ if and only if
    $\CC(G)[N_{\CC(G)}(v)]$ is a clique.
    Furthermore, we set
    \iflong
    \begin{equation*}
        \overline{A} \coloneqq \bigcup A \text{ and }
        \overline{B} \coloneqq \bigcup B,
      \end{equation*}
    \else
      $\overline{A} \coloneqq \bigcup A$ and $\overline{B} \coloneqq \bigcup B$,
    \fi
    that is, the partition of $V(G)$ induced by $A \cupdot B$.
    If $|A| = |\overline{A}|$ and $\N{G} > 3k$, then $(G, \N{G} + k)$ is a negative instance of \textsc{SCC}.
\end{restatable}
  
  \appendixproof{lemma:three_k_negative}{
\ifshort\threeknegativelemma*\fi
\begin{proof}
    We assume that $|A| = |\overline{A}|$ and $\N{G} > 3k$.
    Let $\mathcal{C}$ be a sigma clique cover of $G$.
    We claim that
    \begin{equation*}
        \wgt(\mathcal{C}) \ge \max \set{2 |\overline{A}|, \N{G} + |\overline{B}|} > \N{G} + k.
    \end{equation*}
    First, we will derive $\wgt(\mathcal{C}) \ge 2 |\overline{A}|$:
    Consider the set $B' \subseteq B$ with
    \begin{equation*}
        B' \coloneqq \set{b \in B \mid N_{\CC(G)}(b) \cap A \neq \emptyset}.
    \end{equation*}
    Phrased differently, $B'$ is the subset of $B$ where each element has at least one neighbor in $A$ in $\CC(G)$.
    Furthermore, let $f_B \colon B \to \overline{B} $ such that $f_B(b) \in b$ for all $b \in B$,
    that is, a function selecting an arbitrary vertex out of each critical clique contained in $B$.
    Additionally, we define a second function $f_A \colon A \to \overline{A}$ in a completely symmetric manner.

    Now, consider some $b \in B'$ and $a \in A$ such that $ab \in E(\CC(G))$.
    By \cref{lemma:green}, there is precisely one $C \in \mathcal{C}$ such that
    $\set{f_A(a), f_B(b)} \subseteq C$. Thus, accounting for all such $a$, we obtain
    \begin{equation*}
        \val_\mathcal{C}(f_B(b)) \ge |N_{\CC(G)}(b) \cap A|.
    \end{equation*}
    On the other hand, let $a \in A$.
    Suppose $N_{\CC(G)}(a) \subseteq A$. Then,
    \begin{equation*}
        G \left [\bigcup (N_{\CC(G)}(a) \cup \set{a}) \right ]
    \end{equation*}
    is a connected component of $G$ that is a clique, which we required to never be the case.
    Thus $|N_{\CC(G)}(a) \cap B'| \ge 1$.
    Using these two facts, we obtain
    \ArrowsC{9.0cm}{
        \sum_{b \in B'} \val_\mathcal{C}(f_B(b)) &\ge \sum_{b \in B'} |N_{\CC(G)}(b) \cap A|  \Arrow{double counting principle} \\
                                                 &=   \sum_{a \in A} |N_{\CC(G)}(a) \cap B'| \\
                                                 &\ge |A|\\
                                                 &= |\overline{A}|.
    }
    In total, we calculate%
    \ArrowsC{11.6cm}{
        \wgt(\mathcal{C}) &= \sum_{\overline{v} \in V(G)} \val_\mathcal{C}(\overline{v}) \\
                          &= \sum_{c \in A \cup (B \setminus B') } \sum_{\overline{c} \in c} \val_\mathcal{C}(\overline{c})
                          + \sum_{b \in B'} \sum_{\overline{b} \in b} \val_\mathcal{C}(\overline{b})
                          \Arrow{${A \subseteq A \cup (B \setminus B')}$, ${\forall \overline{v} \in V(G) \colon \val_\mathcal{C}(\overline{v}) \ge 1}$}\\
                          &\ge |\overline{A}| + \sum_{b \in B'} \sum_{\overline{b} \in b} \val_\mathcal{C}(\overline{b})
                          \Arrow{$f_B(b) \in b$}\\
                          &\ge |\overline{A}| + \sum_{b \in B'} \val_\mathcal{C}(f_B(b))\\
                          &\ge 2|\overline{A}|.
    }

    Next, we will derive $\wgt(\mathcal{C}) \ge \N{G} + |\overline{B}|$:
    Let $[v]_G \in B$. By definition of $B$, there are distinct
    $[u]_G, [w]_G \in V(\CC(G))$ such that $\set{vu, vw} \subseteq E(G)$, but
    $uw \not\in E(G)$.
    Let $C_1 \in \mathcal{C}$ such that $\set{v, u} \subseteq C_1$
    and $C_2 \in \mathcal{C}$ such that $\set{v, w} \subseteq C_2$.
    Since $uw \not\in E(G)$, we know that $C_1$ differs from $C_2$. Thus, $\val_\mathcal{C}(v) \ge 2$.
    In total, we obtain
    \ArrowsC{9.65cm}{
        \wgt(\mathcal{C}) &= \sum_{\overline{v} \in V(G)} \val_\mathcal{C}(\overline{v}) \\
                                    &= \sum_{b \in B} \sum_{\overline{b} \in b} \val_\mathcal{C}(\overline{b})
                                    +  \sum_{a \in A} \sum_{\overline{a} \in a} \val_\mathcal{C}(\overline{a})\\
                                    &\ge 2|\overline{B}| + |\overline{A}| \Arrow{$\N{G} = |\overline{A}| + |\overline{B}|$}\\
                                    &= \N{G} + |\overline{B}|.
    }
    To finish our proof, we will combine these two bounds to obtain that $\wgt(\mathcal{C}) > \N{G} + k$.
    First, suppose that $|\overline{A}| \ge \frac{2}{3}\N{G}$. Then,
    \ArrowsC{5.6cm}{
        \wgt(\mathcal{C}) &\ge 2 |\overline{A}|\\
                          &\ge \frac{4}{3} \N{G} \Arrow{$\N{G} > 3k$}\\
                          &> \N{G} + k.
    }
    If otherwise $|\overline{A}| < \frac{2}{3}\N{G}$, then
    \ArrowsC{6.95cm}{
        \wgt(\mathcal{C}) &\ge \N{G} + |\overline{B}| \Arrow{$|\overline{B}| \ge \frac{1}{3}\N{G}$}\\
                          &\ge \N{G} + \frac{1}{3}\N{G} \Arrow{$\N{G} > 3k$}\\
                          &> \N{G} + k.
    }
    Therefore, we conclude that $\wgt(\mathcal{C}) > \N{G} + k$ in all cases.
    Since $\mathcal{C}$ was chosen generically, this implies $(G, \N{G} + k)$ is a negative instance of \textsc{SCC}.
\qed
\end{proof}
}                               %

\subsection{Deriving the Kernel}
\label{section:kernel}

In the two preceding sections, we essentially derived two data reduction rules for the sigma clique cover problem.
It remains to compile our results into a polynomial kernelization procedure for \textsc{Cluster Vertex Splitting}.
Essentially, we convert a given instance $(G, k)$ of \textsc{Cluster Vertex Splitting} into
an equivalent instance of \textsc{Sigma Clique Cover}, apply the two reduction rules exhaustively, until finally converting the reduced instance back to an instance of \textsc{Cluster Vertex Splitting}.
Refer to \cref{figure:cluster_kernel_example} for an example.

\begin{theorem}\label{theorem:cluster_kernel}
    \textsc{Cluster Vertex Splitting} admits a problem kernelization mapping an instance $(G, k)$ to an equivalent instance $(G', k')$ satisfying
    $\N{G'} \leq 3k + 3$ and $k' \leq k$.
  \end{theorem}
  \gappto{\appendixText}{\subsection{Missing Parts of the Proof of \cref{theorem:cluster_kernel}}}
\begin{proof}
    Let an instance of \textsc{Cluster Vertex Splitting} be given through $(G, k)$ and let $G_0$ be obtained from $G$ by removing all isolated vertices.
    Observe that $(G, k)$ is equivalent to $(G_0, k \eqqcolon k_0)$ with respect to \textsc{CVS}.
    \iflong We apply \cref{lemma:cvs_scc_reduction} and derive that $(G_0, k_0)$ is a positive instance of \textsc{CVS} if and only if $(G_0, \N{G_0} + k_0)$ is a positive instance of \textsc{Sigma Clique Cover}.\else
    By the equivalence of \textsc{CVS} and \textsc{Sigma Clique Cover}, $(G_0, k_0)$ is a positive instance of \textsc{CVS} if and only if $(G_0, \N{G_0} + k_0)$ is a positive instance of \textsc{Sigma Clique Cover}.
    \fi
    Next, we construct the sequences $G_0, \ldots$ and $k_0, \ldots$
    by exhaustively applying the following set of rules:

    \begin{figure}[t]
        \begin{center}
            \includegraphics{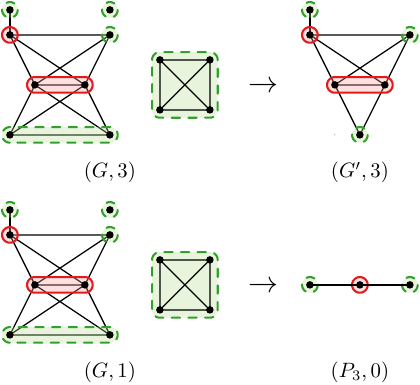}
        \end{center}
        \caption[The kernel for \textsc{CVS} exemplified for two concrete instances.]{Two instances of \textsc{CVS} and their corresponding kernel as given by \cref{theorem:cluster_kernel}.
        Reducible critical cliques are marked in green with dashed outlines, while irreducible critical cliques are marked in red with solid outlines.
        \iflong Note that the graph $G$ is taken from \cref{figure:cc_example}.\fi }
        \label{figure:cluster_kernel_example}
    \end{figure}
    \begin{description}        
        \item[Rule~I:] 
        If there is a critical clique
        $[v]_{G_i} \in V(\CC({G_i}))$ such that $[v]_{G_i}$ contains at least two vertices
        and $\CC(G_i)[N_{\CC(G_i)}([v]_{G_i})]$ is a clique,
        then $G_{i+1} \coloneqq (G_i - v) - I$ and $k_{i+1} \coloneqq k_i$,
        where $I$ is the set of isolated vertices in $G_i - v$.

        \item[Rule~II:]
        If Rule~I is not applicable to $G_i$, Rule~II has not been used so far,
        and $\N{G_i} > 3k_i$, then $G_{i+1} \coloneqq P_3$ and $k_{i+1} \coloneqq 0$. 
    \end{description}

    \ifshort
      The running time and kernel size are proven in the full version of this paper \cite{DBLP:journals/corr/abs-2309-00504}.%
    \fi%
    \toappendix{
    \smallskip\emph{Termination in polynomial time.}
    Observe that Rule~I reduces the number of vertices of the current graph,
    and that Rule~II is applicable at most once.
    Thus, both sequences are finite and of length $\ell = \mathcal{O}(\N{G})$.
    The time complexity of constructing the critical clique graph of some graph $H$ is in $\mathcal{O}(\N{H} + \M{H})$~\cite{4kclustereditingkernel}.
    Hence, we observe that our sequences can be constructed using a budget of $\mathcal{O}(\N{G} (\N{G} + \M{G}) )$ steps.

  }%
  \iflong%
    \smallskip\emph{Correctness.}
  \fi%
    We claim that Rule~I and Rule~II are \emph{correct}, that is, the instances $(G_i, \N{G_i} + k_i)$ and $(G_{i+1}, \N{G_{i+1}} + k_{i+1})$ are equivalent with respect to the \textsc{SCC} problem
    for all $i \in \set{0, \dots, \ell - 1}$.
    Let $G_i$ such that $G_{i+1}$ was obtained by applying Rule~I, and let $v$ as well as $I$ as used in the definition of Rule~I.
    First, consider the case when $I \neq \emptyset$. Let $w \in I$.
    We have that $d_{G_i}(w) \ge 1$, because $w$ is not isolated in $G_i$.
    At the same time, we know that $d_{G_i}(w) < 2$, for otherwise $w$ would not be isolated in $G_i - v$.
    Thus, $d_{G_i}(w) = 1$, which forces $|[v]_{G_i}| = 2$.
    Since $w \not\in [v]_{G_i}$ would imply $d_{G_i}(w) \ge 2$,
    we conclude that $[v]_{G_i} = \set{v, w}$, that is, $G_i[\set{v, w}] \simeq K_2$
    is a connected component of $G_i$.
    Now, it is easy to see that $(G_i, \N{G_i} + k)$ is equivalent 
    to  $( G_{i+1}, \N{G_{i+1}} + k)$
    with respect to the \textsc{SCC} problem.
    Otherwise, $I = \emptyset$.
    By construction, $G_i$ is free of isolated vertices.
    Thus, applying \cref{lemma:shrink_green} yields that $(G_i, \N{G_i} + k)$ is equivalent to $(G_i - v, \N{G_i - v} + k) = (G_{i+1}, \N{G_{i+1}} + k)$  with respect to the \textsc{SCC} problem.
    Hence, Rule~I is correct.

    \ifshort
      The correctness of Rule~II essentially follows from \cref{lemma:three_k_negative}; the proof is given in the full version~\cite{DBLP:journals/corr/abs-2309-00504}.%
    \fi%
    \toappendix{%
      \ifshort%
        \smallskip\emph{Correctness of Rule II.}
      \fi%
    Next, let $G_i$ such that $G_{i+1}$ was obtained by applying Rule~II,
    and let $A, \overline{A}, B, \overline{B}$ as defined in the header of \cref{lemma:three_k_negative} when substituting $G$ for $G_i$. 
    Then, $\N{G_i} > 3k_i$ and Rule~I is not applicable to $G_i$.
    Hence, for all 
    $[v]_{G_i} \in V(\CC({G_i}))$ such that
    $\CC(G_i)[N_{\CC(G_i)}([v]_{G_i})]$ is a clique,
    we have $|[v]_{G_i}| = 1$. Note that this implies $|A| = |\overline{A}|$. 
    Again, notice also that $G_i$ cannot contain isolated vertices.
    Now, suppose that $C \subseteq V(G_i)$ induces a connected component of $G_i$ that is a clique
    with $|C| > 1$ and let $v \in C$. Then, $C$ ``spans'' the whole of $G_i[C]$, that is, $[v]_{G_i} = C$ and $\CC(G_i)[N_{\CC(G_i)}([v]_{G_i})] = \emptyset$.
    Thus, applying the above, we have $|C| = 1$, which cannot be since $G_i$ is free of isolated vertices.
    Therefore, none of $G_i$'s connected components are cliques.
    Hence, all conditions are met to apply \cref{lemma:three_k_negative} to $G_i, k_i, A, \overline{A}, B$ and  $\overline{B}$, showing that $(G_i, k_i)$ is a negative instance of \textsc{SCC}.
    \iflong%
      Since $(P_3, |P_3| + 0)$ is a negative instance of \textsc{SCC} too, Rule~II is correct.
      
    \else%
      Since $(P_3, |P_3| + 0)$ is negative, too, Rule~II is correct.
      
    \fi%
    }%
    In total, we have that $(G_\ell, \N{G_\ell} + k_\ell)$ is a positive instance of \textsc{SCC} if and only if $(G_0, \N{G_0} + k_0)$ is.     
    \iflong
      Another application of \cref{lemma:cvs_scc_reduction}
      (using that $G_\ell$ is free of isolated vertices)
      yields that $(G_\ell, \N{G_\ell} + k_\ell)$ is a positive instance of \textsc{SCC} if and only if $(G_\ell, k_\ell)$ is a positive instance of \textsc{CVS}.
    \else
      By the equivalence between \textsc{CVS} and \textsc{Sigma Clique Cover}
      $(G_\ell, \N{G_\ell} + k_\ell)$ is a positive instance of \textsc{SCC} if and only if $(G_\ell, k_\ell)$ is a positive instance of \textsc{CVS}.
    \fi
    Finally, we conclude that $(G, k)$ is equivalent to $(G_\ell, k_\ell)$ with respect to the \textsc{CVS} problem.
    \toappendix{

    \smallskip\emph{Problem kernel size.}
    First, we observe that $k_\ell \leq k$, as no rule may increase the current value for $k$.
    If Rule~II was used in the construction of the sequence at any step, then $\N{G_\ell} = |P_3| \leq 3 + k$.
    Otherwise, Rule~II was not used.
    As $G_\ell$ is the last element of $G_0, \dots, G_\ell$, no rule is applicable to it.
    Suppose $\N{G_\ell} > 3k_\ell$.
    But then, Rule~II is applicable, which is a contradiction.
    Hence, $\N{G_\ell} \leq 3k_\ell \leq 3k + 3$.
    }%
\qed
\end{proof}

\section{The Critical-Clique Lemma}
\label{sec:ccl}
\appendixsection{sec:ccl}

We now consider the critical-clique lemma for \pCEVSlong\ (\pCEVS) mentioned in the introduction.
\pCEVS\ is defined as follows, where by a \emph{graph modification} we mean a vertex split, an edge addition, or an edge deletion.

\begin{problembox}[framed]{\problempad \pCEVSlong}
  \problempad Input: & A tuple $(G, k)$, where $G$ is a graph and $k \in \mathbb{N}$. \\
  \problempad Question: & Is there a sequence of at most $k$ graph modifications that transforms $G$ into a cluster graph?
\end{problembox}

To state the critical-clique lemma for \pCEVS, we first need an equivalence between the sequence of modifications in \pCEVS\ and a cover of the input graph by clusters\iflong, similar to the correspondence between sigma clique covers and cluster vertex splittings in \cref{lemma:cvs_scc_reduction}\fi.

A \emph{cover} of a graph~$G$ is a collection~$\mathcal{C}$ of subsets of $V(G)$ such that $\bigcup_{C \in \mathcal{C}} C = V(G)$.
The \emph{cost}~$\cst_G(\cC)$ of a cover~\cC\ is the number of non-edges contained in a set of~\cC\ plus the number of edges not contained in any set of~\cC\ plus the number of times each vertex is covered by a set beyond the first time.
In formulas,
\begin{multline*}
  \cst_G(\cC) =
  \left|\left\{uv \in \binom{V}{2}\setminus E(G) \mid \exists C \in \cC \colon \{u, v\} \subseteq C\right\}\right| + {}\\
  \left|\{uv \in E(G) \mid \forall C \in \cC \colon \{u, v\} \nsubseteq C\}\right| +
  \left(\sum_{C \in \cC}|C|\right) - |V(G)|.
\end{multline*} %
Herein, $\binom{V}{2}$ denotes the set of all two-element subsets of $V$.
If $G$ is clear from the context, we omit the subscript~$G$ in $\cst_G$.

The following lemma has been used implicitly by Abu-Khzam et al.~\cite{abu-khzam_cluster_2018} but we are not aware of a formal proof.
\begin{restatable}[\appsymb]{lemma}{cescoverequiv}%
  \label{lem:ces-cover-equiv}%
  Let $G$ be a graph and $k$ a positive integer.
  There is a sequence of at most $k$ graph modifications to obtain from~$G$ a cluster graph if and only if $G$ admits a cover of cost at most~$k$.
\end{restatable}
\appendixproof{lem:ces-cover-equiv}{
\ifshort\cescoverequiv*\fi
\begin{proof}
  Let $S$ be a sequence of at most $k$ graph modifications such that applying them to $G$ results in a cluster graph.
  By a reordering argument of Abu-Khzam et al.~\cite{abu-khzam_cluster_2018} (see \cite[Theorem 1]{abu-khzam_cluster_2019v1}) we may assume that $S$ consists of a possibly empty sequence of edge additions, then a possibly empty sequence of edge deletions, and then a possibly empty sequence of vertex splits.
  Consider the graph $\tilde{G}$ obtained after performing all edge additions and edge deletions but none of the vertex splits.
  Let $\ell$ be the number of vertex splits in $S$ and $n_0$ the number of degree-0 vertices in $\tilde{G}$.
  By \cref{lemma:cvs_scc_reduction} there is a sigma clique cover of $\tilde{G}$ of weight at most $n - n_0 + \ell$ were $n$ is the number of vertices of $\tilde{G}$.
  By adding to this sigma clique cover the degree-0 vertices of $\tilde{G}$ as singleton sets, we obtain a cover $\cC$ of $\tilde{G}$.
  Observe that the cost of \cC\ (with respect to $\tilde{G}$) is at most $\ell$%
  .
  Notice that the number of edges of $G$ that are not contained in any set in $\cC$ is at most the number of edge deletions in $S$ and that the number of non-edges of $G$ that are contained in at least one set in $\cC$ is at most the number of edge additions in $S$.
  Hence, $\cC$ is a cover of $G$ of cost at most $k$.

  Now let $\cC$ be a cover of $G$ of cost at most $k$.
  Delete each edge from $G$ that is not in any set in $\cC$ and for each non-edge of $G$ that is contained in some set of $\cC$, add the corresponding edge to $G$.
  Denote by $\tilde{G}$ the so-obtained graph.
  Let $k'$ be obtained from $k$ by subtracting the number of performed graph modifications so far.
  Note that \cC\ is a sigma clique cover of~$\tilde{G}$.
  Remove the isolated vertices from~\cC, obtaining $\cC'$, which is still a sigma clique cover of~$\tilde{G}$.
  Moreover, the weight of~$\cC'$ with respect to $\tilde{G}$ is at most $n - n_0 + k'$, where $n_0$ is the number of isolated vertices in $\tilde{G}$, by the definition of the cost of $\cC$.
  Thus, by \cref{lemma:cvs_scc_reduction} we may split at most $k'$ vertices in $\tilde{G}$ to obtain a cluster graph.
\qed
\end{proof}
} %

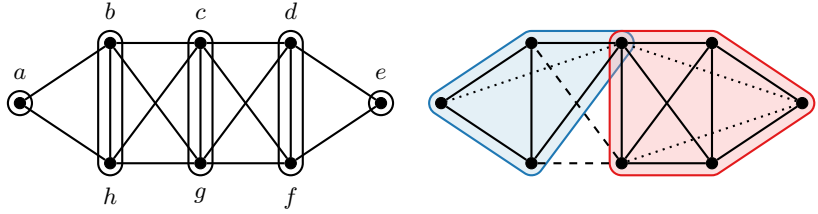
\begin{figure}[t]
  \centering
  \begin{tikzpicture}[
    vert/.style={circle,inner sep=1.5pt,draw=black,fill=black},
    every edge/.style={draw,thick},
    hedge/.style={thick,fill=#1,fill opacity=0.3},
    label distance=1.25mm,
    scale=0.8]
    \input{graphics/tikz-hypergraph.tex}

    \node[vert,label=above:$a$] (a) at (0,-1) {};
    \node[vert,label=above:$b$] (b) at (1.5,0) {}
    edge (a);
    \node[vert,label=above:$c$] (c) at (3,0) {}
    edge (b);
    \node[vert,label=above:$d$] (d) at (4.5,0) {}
    edge (c);
    \node[vert,label=above:$e$] (e) at (6,-1) {}
    edge (d);
    \node[vert,label=below:$f$] (f) at (4.5,-2) {}
    edge (e)
    edge (d)
    edge (c);
    \node[vert,label=below:$g$] (g) at (3,-2) {}
    edge (f)
    edge (d)
    edge (c)
    edge (b);
    \node[vert,label=below:$h$] (h) at (1.5,-2) {}
    edge (g)
    edge (c)
    edge (b)
    edge (a);

    \begin{pgfonlayer}{bg}
      \draw[hedge=white] \hedgei{a}{2mm};
      \draw[hedge=white] \hedgeii{b}{h}{2mm};
      \draw[hedge=white] \hedgeii{c}{g}{2mm};
      \draw[hedge=white] \hedgeii{d}{f}{2mm};
      \draw[hedge=white] \hedgei{e}{2mm};
    \end{pgfonlayer}
    \begin{scope}[xshift=7cm]
      \node[vert] (a) at (0,-1) {};
      \node[vert] (b) at (1.5,0) {}
      edge (a);
      \node[vert] (c) at (3,0) {}
      edge (b)
      edge[dotted] (a);
      \node[vert] (d) at (4.5,0) {}
      edge (c);
      \node[vert] (e) at (6,-1) {}
      edge (d)
      edge[dotted] (c);
      \node[vert] (f) at (4.5,-2) {}
      edge (e)
      edge (d)
      edge (c);
      \node[vert] (g) at (3,-2) {}
      edge (f)
      edge (d)
      edge (c)
      edge[dashed] (b)
      edge[dotted] (e);
      \node[vert] (h) at (1.5,-2) {}
      edge[dashed] (g)
      edge (c)
      edge (b)
      edge (a);

      \begin{pgfonlayer}{bg}
        \draw[hedge=bluex,draw=darkbluex] \hedgem{a}{b}{c,h}{2mm};
        \draw[hedge=red,draw=darkred] \hedgem{c}{d}{e,f,g}{2mm};
      \end{pgfonlayer}
    \end{scope}
  \end{tikzpicture}

  \caption{Counterexample to the critical-clique lemma.}\label{fig:criticalcliqueerror}
\end{figure}

\toappendix{
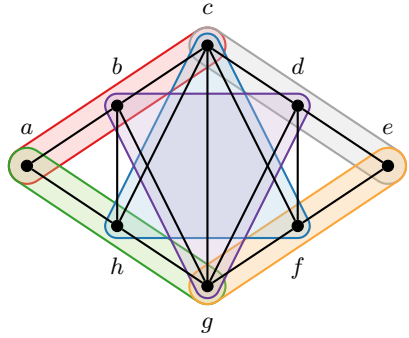
\begin{figure}[t]
  \centering
  \begin{tikzpicture}[
    vert/.style={circle,inner sep=1.5pt,draw=black,fill=black},
    every edge/.style={draw,thick},
    hedge/.style={thick,fill=#1,fill opacity=0.3},
    scale=0.8,
    label distance=2.25mm]
    \input{graphics/tikz-hypergraph.tex}

    \node[vert,label=above:$a$] (a) at (0,-1) {};
    \node[vert,label=above:$b$] (b) at (1.5,0) {}
    edge (a);
    \node[vert,label=above:$c$] (c) at (3,1) {}
    edge (b);
    \node[vert,label=above:$d$] (d) at (4.5,0) {}
    edge (c);
    \node[vert,label=above:$e$] (e) at (6,-1) {}
    edge (d);
    \node[vert,label=below:$f$] (f) at (4.5,-2) {}
    edge (e)
    edge (d)
    edge (c);
    \node[vert,label=below:$g$] (g) at (3,-3) {}
    edge (f)
    edge (d)
    edge (c)
    edge (b);
    \node[vert,label=below:$h$] (h) at (1.5,-2) {}
    edge (g)
    edge (c)
    edge (b)
    edge (a);

    \begin{pgfonlayer}{bg}
      \draw[hedge=red,draw=darkred] \hedgem{a}{b}{c}{3mm};
      \draw[hedge=lightgray,draw=gray] \hedgem{c}{d}{e}{3mm};
      \draw[hedge=green,draw=darkgreen] \hedgem{a}{g}{h}{3mm};
      \draw[hedge=brown,draw=darkbrown] \hedgem{e}{f}{g}{3mm};
      \draw[hedge=bluex,draw=darkbluex] \hedgem{h}{c}{f}{2mm};
      \draw[hedge=violet,draw=darkviolet] \hedgem{d}{g}{b}{2mm};
    \end{pgfonlayer}
  \end{tikzpicture}

  \caption{The $P_3$ packing in \cref{prop:ccl-counterexample}.}\label{fig:ccl-packing}
\end{figure}
}                               %

\toappendix{
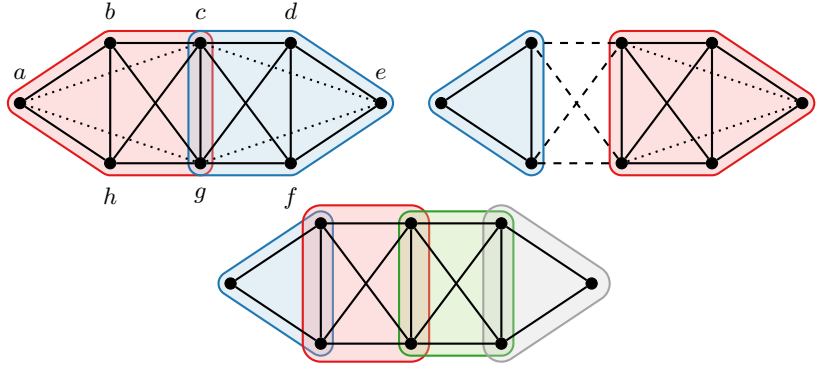
\begin{figure}[t]
  \centering
  \begin{tikzpicture}[
    vert/.style={circle,inner sep=1.5pt,draw=black,fill=black},
    every edge/.style={draw,thick},
    hedge/.style={thick,fill=#1,fill opacity=0.3},
    scale=0.8,
    label distance=1.25mm]
    \input{graphics/tikz-hypergraph.tex}

    \node[vert,label=above:$a$] (a) at (0,-1) {};
    \node[vert,label=above:$b$] (b) at (1.5,0) {}
    edge (a);
    \node[vert,label=above:$c$] (c) at (3,0) {}
    edge (b)
    edge[dotted] (a);
    \node[vert,label=above:$d$] (d) at (4.5,0) {}
    edge (c);
    \node[vert,label=above:$e$] (e) at (6,-1) {}
    edge (d)
    edge[dotted] (c);
    \node[vert,label=below:$f$] (f) at (4.5,-2) {}
    edge (e)
    edge (d)
    edge (c);
    \node[vert,label=below:$g$] (g) at (3,-2) {}
    edge (f)
    edge (d)
    edge (c)
    edge (b)
    edge[dotted] (e)
    edge[dotted] (a);
    \node[vert,label=below:$h$] (h) at (1.5,-2) {}
    edge (g)
    edge (c)
    edge (b)
    edge (a);

    \begin{pgfonlayer}{bg}
      \draw[hedge=red,draw=darkred] \hedgem{a}{b}{c,g,h}{2mm};
      \draw[hedge=bluex,draw=darkbluex] \hedgem{c}{d}{e,f,g}{2mm};
    \end{pgfonlayer}

    \begin{scope}[xshift=7cm]
      \node[vert] (a) at (0,-1) {};
      \node[vert] (b) at (1.5,0) {}
      edge (a);
      \node[vert] (c) at (3,0) {}
      edge[dashed] (b);
      \node[vert] (d) at (4.5,0) {}
      edge (c);
      \node[vert] (e) at (6,-1) {}
      edge (d)
      edge[dotted] (c);
      \node[vert] (f) at (4.5,-2) {}
      edge (e)
      edge (d)
      edge (c);
      \node[vert] (g) at (3,-2) {}
      edge (f)
      edge (d)
      edge (c)
      edge[dashed] (b)
      edge[dotted] (e);
      \node[vert] (h) at (1.5,-2) {}
      edge[dashed] (g)
      edge[dashed] (c)
      edge (b)
      edge (a);

      \begin{pgfonlayer}{bg}
        \draw[hedge=bluex,draw=darkbluex] \hedgem{a}{b}{h}{2mm};
        \draw[hedge=red,draw=darkred] \hedgem{c}{d}{e,f,g}{2mm};
      \end{pgfonlayer}
    \end{scope}

    \begin{scope}[xshift=3.5cm,yshift=-3cm]
      \node[vert] (a) at (0,-1) {};
      \node[vert] (b) at (1.5,0) {}
      edge (a);
      \node[vert] (c) at (3,0) {}
      edge (b);
      \node[vert] (d) at (4.5,0) {}
      edge (c);
      \node[vert] (e) at (6,-1) {}
      edge (d);
      \node[vert] (f) at (4.5,-2) {}
      edge (e)
      edge (d)
      edge (c);
      \node[vert] (g) at (3,-2) {}
      edge (f)
      edge (d)
      edge (c)
      edge (b);
      \node[vert] (h) at (1.5,-2) {}
      edge (g)
      edge (c)
      edge (b)
      edge (a);

      \begin{pgfonlayer}{bg}
        \draw[hedge=bluex,draw=darkbluex] \hedgem{a}{b}{h}{2mm};
        \draw[hedge=red,draw=darkred] \hedgem{b}{c}{g,h}{3mm};
        \draw[hedge=green,draw=darkgreen] \hedgem{c}{d}{f,g}{2mm};
        \draw[hedge=lightgray,draw=gray] \hedgem{d}{e}{f}{3mm};
      \end{pgfonlayer}
    \end{scope}
    
  \end{tikzpicture}

  \caption{Solutions of cost 6 that do not cut critical cliques.}\label{fig:ccl-more-solns}
\end{figure}
}

Recall the definition of critical cliques from \cref{definition:cc}.
The critical-clique lemma as stated by Abu-Khzam et al.~\cite{abu-khzam_cluster_2018} (see their Lemma~8) is as follows:

\begin{lemma}[Incorrect]
  \label{lem:abu-ccl}
  Let $G$ be a graph and \cC\ a cover of $G$ of minimum cost.
  For each $C \in \cC$ and each critical clique of $G$ with vertex set $K$ we have either $C \cap K = \emptyset$ or $K \subseteq C$.
\end{lemma}
As far as we are aware, \cref{lem:abu-ccl} is being used in references \cite{ArrighiBDSW23,abu-khzam_cluster_2019v1,abu-khzam_cluster_2018,askeland_overlapping_2022}. However, the example in \cref{fig:criticalcliqueerror} shows that \cref{lem:abu-ccl} is incorrect: The left shows the input graph with marked critical cliques.
The right shows a minimum-cost cover in which the left cover set contains the central critical clique only partially.
Dashed edges are removed, dotted edges added, and vertices in both sets are split.
The cost of the cover is 6.

\begin{proposition}
  \label{prop:ccl-counterexample}
  The graph shown on the left in \cref{fig:criticalcliqueerror} needs at least 6 modifications to turn it into a cluster graph.
\end{proposition}
\begin{proof}
  We show that there is a modification-disjoint packing of six induced $P_3$s.
  In the following, we denote a $P_3$ by $xyz$, where $x$, $y$, and $z$ are its three vertices and $y$ is the center vertex.
  \iflong
  Two $P_3$s $xyz$ and $abc$ are \emph{modification disjoint} if they do not contain the same vertex pair (that is, the same edge or non-edge) and they do not contain the same center vertex.
  In formulas, $|\{a, b, c\} \cap \{x, y, z\}| \leq 1$ and $y \neq b$.
\else
  Two $P_3$s $xyz$ and $abc$ are \emph{modification disjoint} if $|\{a, b, c\} \cap \{x, y, z\}| \leq 1$ and $y \neq b$.
  \fi
  A modification-disjoint packing of $P_3$s is a collection of induced $P_3$s that are pairwise modification disjoint.
  Note that, if a graph admits a modification-disjoint packing of $\ell$ $P_3$s then we need at least $\ell$ modifications to turn the graph into a cluster graph.

  Consider the following $P_3$s in the graph in \cref{fig:criticalcliqueerror}: $abc$, $cde$, $ahg$, $gfe$, $hcf$, $bgd$.
  \iflong
    See also \cref{fig:ccl-packing}.
  \fi
  Note that they form a modification-disjoint packing.
  Thus we need at least 6 modifications to turn the graph into a cluster graph.
\qed
\end{proof}
There are other solutions of cost 6 that do not cut critical cliques.
Thus, it is tempting to assume that, altough not necessarily every optimal solution does not cut critical cliques, that there is always such an optimal solution. 
Indeed, after the appearance of our counterexample above, this has been proved to be true:
\begin{lemma}[Abu-Khzam et al.~\cite{abu-khzam_cluster_2023v2}]\label{lem:critical-cliques-repaired}
  Let $G$ be a graph and $k$ a positive integer.
  If $(G, k)$ admits a solution for \pCEVS, then there is a cover~\cC\ of cost at most~$k$ such that for each critical clique $K$ of $G$ and each set $C \in \cC$ we have either $K \subseteq C$ or $K \cap C = \emptyset$.
\end{lemma}

\section{The Complexity of \pCEVSlong}
\label{sec:cevs}
\appendixsection{sec:cevs}
Based mainly on our \NP-hardness proof of \pCVSlong\ in conjunction with the corrected critical-clique lemma we obtain \NP-hardness of \pCEVSlong:

\begin{restatable}[\appsymb]{theorem}{ceshardtheorem}\label{thm:ces-hard}
  There is a polynomial-time many-one reduction from \pCVS\ to \pCEVS, showing that \pCEVS\ is \NP-hard.
\end{restatable}  
\appendixproof{thm:ces-hard}{
\ifshort\ceshardtheorem*\fi
\begin{proof}
  We give a reduction from \pCVSlong~(\pCVS) to \pCEVSlong~(\pCEVS).
  Let $(G, k)$ be an instance of \pCVS.
  Without loss of generality, we assume that $G$ does not contain isolated vertices.
  We construct an instance $(H, s)$ of \pCEVS.
  To obtain~$H$ from $G$, replace each vertex in $G$ by a clique with $k + 1$ vertices.
  That is, $V(H) = \{v_i \mid v \in V(G), i \in [k + 1]\}$ and $E(H) = \{u_iv_j \mid uv \in E(G), i, j \in [k + 1]\}$.
  We say that $v_i \in V(H)$ is a \emph{copy} of $v \in V(G)$ and for each $v \in V(G)$ we let $K_v := \{v_i \in V(H) \mid i \in [k + 1]\}$ denote the \emph{clique of~$v$}.
  Put $s = k(k + 1)$.
  Clearly, the reduction can be carried out in polynomial time.
  It remains to prove that $(G, k)$ has a solution (for \pCVS) if and only if $(H, s)$ has a solution (for \pCEVS).
  
  Let $S$ be a solution to $(G, k)$.
  By \cref{lemma:cvs_scc_reduction} there is a sigma clique cover~\cC\ for $G$ of weight at most~$n + k$.
  From $\cC$, construct a cover $\cC'$ for $H$ by replacing in each set of \cC\ each vertex by all of its copies.
  That is $\cC' = \{ \{v_i \mid v \in C, i \in [k + 1]\} \mid C \in \cC\}$.
  Observe that, since each set in \cC\ is a clique with $k + 1$ vertices, we have $\cst(\cC') \leq k(k + 1)$.
  Thus, $(H, s)$ has a solution by \cref{lem:ces-cover-equiv}.

  Let $S$ be a solution to $(H, s)$.
  By \cref{lem:ces-cover-equiv} there is a cover~$\cC'$ of cost at most~$k(k + 1)$.
  By \cref{lem:critical-cliques-repaired} we may assume that $\cC'$ is such that for each critical clique in $H$ with vertex set $K$ and each set $C' \in \cC'$ we have either $K \subseteq C'$ or $K \cap C' = \emptyset$.
  We claim that $\cC'$ is a sigma clique cover for~$H$.
  Observe that for each $v \in V(G)$ we have that $K_v$ is contained in some critical clique of~$H$.
  Hence, for all $C' \in \cC'$ we have either $K_v \subseteq C'$ or $K_v \cap C' = \emptyset$.
  We claim that each edge of $H$ is contained in a set of~$\cC'$.
  For a contradiction, assume the contrary, that is, there are $i, j \in [k + 1]$ and $uv \in E(G)$ such that $u_iv_j \in E(H)$ is not contained in any set of~$\cC'$.
  It follows that indeed for all $i, j \in [k + 1]$ we have that $u_iv_j \in E(H)$ is not contained in any set of~$\cC'$.
  That is, the cost of $\cC'$ is at least $(k + 1)^2$, a contradiction to the fact that $\cC'$ has cost at most~$k(k + 1)$.
  Analogously we can show that no non-edge of $H$ is contained in a set of $\cC'$.
  Hence, indeed $\cC'$ is an edge clique cover of~$H$.
  Construct a sigma clique cover \cC\ for $G$ by replacing each clique $K_v$ by $v$, that is, put $\cC = \{\{v \in V(G) \mid K_v \subseteq C'\} \mid C' \in \cC'\}$.
  Observe that \cC\ has weight at most~$n + k$.
  Thus, by \cref{lemma:cvs_scc_reduction}, $(G, k)$ has a solution, as required.
\qed
\end{proof}
}                               %

\section{Conclusion}

We conclude with directions for future research.
The constants in our kernelization for \pCVS\ (at most $3k + 3$ vertices, see \cref{theorem:cluster_kernel}) are already quite small, but it would be interesting to see whether they can be further improved.
A problem kernel with a linear number of edges would also be interesting.
A straightforward brute-force search on the kernel yields an algorithm solving \pCVS\ in $2^{O(k^2)} \cdot n^{O(1)}$ time, which can be improved to $2^{O(k \log k)} \cdot n^{O(1)}$ with further observations.
Is it possible to obtain $2^{O(k)} \cdot n^{O(1)}$ time as well?
Finally, we focused here on the case where the overlap between clusters is small.
There are applications where the overlap is relatively large~\cite{yang_structure_2014}.
Thus, to get efficient algorithms in this case, it would be interesting to study parameterizations dual to $k$ that measure the non-overlapping parts of the clustering.

 \iflong\emergencystretch=1em\fi %
\printbibliography

\end{document}

%% file: graphics/tikz-hypergraph.tex
\def\rotateclockwise#1{
  \newdimen\xrw
  \pgfextractx{\xrw}{#1}
  \newdimen\yrw
  \pgfextracty{\yrw}{#1}
  \pgfpoint{\yrw}{-\xrw}
}

\def\rotatecounterclockwise#1{
  \newdimen\xrcw
  \pgfextractx{\xrcw}{#1}
  \newdimen\yrcw
  \pgfextracty{\yrcw}{#1}
  \pgfpoint{-\yrcw}{\xrcw}
}

\def\outsidespacerpgfclockwise#1#2#3{
  \pgfpointscale{#3}{
    \rotateclockwise{
      \pgfpointnormalised{
        \pgfpointdiff{#1}{#2}}}}
}

\def\outsidespacerpgfcounterclockwise#1#2#3{
  \pgfpointscale{#3}{
    \rotatecounterclockwise{
      \pgfpointnormalised{
        \pgfpointdiff{#1}{#2}}}}
}

\def\outsidepgfclockwise#1#2#3{
  \pgfpointadd{#2}{\outsidespacerpgfclockwise{#1}{#2}{#3}}
}

\def\outsidepgfcounterclockwise#1#2#3{
  \pgfpointadd{#2}{\outsidespacerpgfcounterclockwise{#1}{#2}{#3}}
}

\def\outside#1#2#3{
  ($ (#2) ! #3 ! -90 : (#1) $)
}

\def\cornerpgf#1#2#3#4{
  \pgfextra{
    \pgfmathanglebetweenpoints{#2}{\outsidepgfcounterclockwise{#1}{#2}{#4}}
    \let\anglea\pgfmathresult
    \let\startangle\pgfmathresult

    \pgfmathanglebetweenpoints{#2}{\outsidepgfclockwise{#3}{#2}{#4}}
    \pgfmathparse{\pgfmathresult - \anglea}
    \pgfmathroundto{\pgfmathresult}
    \let\arcangle\pgfmathresult
    \ifthenelse{180=\arcangle \or 180<\arcangle}{
      \pgfmathparse{-360 + \arcangle}}{
      \pgfmathparse{\arcangle}}
    \let\deltaangle\pgfmathresult

    \newdimen\x
    \pgfextractx{\x}{\outsidepgfcounterclockwise{#1}{#2}{#4}}
    \newdimen\y
    \pgfextracty{\y}{\outsidepgfcounterclockwise{#1}{#2}{#4}}
  }
  -- (\x,\y) arc [start angle=\startangle, delta angle=\deltaangle, radius=#4]
}

\def\corner#1#2#3#4{
  \cornerpgf{\pgfpointanchor{#1}{center}}{\pgfpointanchor{#2}{center}}{\pgfpointanchor{#3}{center}}{#4}
}

\def\hedgeiii#1#2#3#4{
  \outside{#1}{#2}{#4} \corner{#1}{#2}{#3}{#4} \corner{#2}{#3}{#1}{#4} \corner{#3}{#1}{#2}{#4} -- cycle
}

\def\hedgem#1#2#3#4{

  \outside{#1}{#2}{#4}
  \pgfextra{
    \def\hgnodea{#1}
    \def\hgnodeb{#2}
  }
  foreach \c in {#3} {
    \corner{\hgnodea}{\hgnodeb}{\c}{#4}
    \pgfextra{
      \global\let\hgnodea\hgnodeb
      \global\let\hgnodeb\c
    }
  }
  \corner{\hgnodea}{\hgnodeb}{#1}{#4}
  \corner{\hgnodeb}{#1}{#2}{#4}
  -- cycle
}

\def\hedgeii#1#2#3{
  \hedgem{#1}{#2}{}{#3}
}

\def\hedgei#1#2{
  (#1) circle [radius = #2]
}